%% file: main.tex
\documentclass[fullpage,11pt]{article}

\usepackage{amsthm}
\usepackage{graphicx} 
\usepackage{array} 

\usepackage{amsmath, amssymb, amsfonts, verbatim}
\usepackage{hyphenat,epsfig,subfigure,multirow}

\usepackage[usenames,dvipsnames]{xcolor}
\usepackage[ruled]{algorithm2e}

\usepackage{tcolorbox}

\definecolor{DarkRed}{rgb}{0.5,0.1,0.1}
\definecolor{DarkBlue}{rgb}{0.1,0.1,0.5}

\usepackage{hyperref}
\hypersetup{
colorlinks=true,
pdfnewwindow=true,
citecolor=ForestGreen,
linkcolor=DarkRed,
filecolor=DarkRed,
urlcolor=DarkBlue
}

\usepackage{bm}
\usepackage{url}
\usepackage{xspace} 
\usepackage[mathscr]{euscript}

\usepackage{mdframed}

\usepackage[noend]{algpseudocode}
\makeatletter
\def\BState{\State\hskip-\ALG@thistlm}
\makeatother

\usepackage{cite}
\usepackage{enumerate}

\usepackage[margin=1in]{geometry}

\newtheorem{theorem}{Theorem}
\newtheorem{lemma}{Lemma}[section]
\newtheorem{proposition}[lemma]{Proposition}

\newtheorem{claim}[lemma]{Claim}
\newtheorem{fact}[lemma]{Fact}

\newtheorem{definition}{Definition}
\newtheorem{problem}{Problem}
\newtheorem{remark}[lemma]{Remark}
\newtheorem{observation}[lemma]{Observation}

\newtheorem*{claim*}{Claim}
\newtheorem*{proposition*}{Proposition}
\newtheorem*{lemma*}{Lemma}
\newtheorem*{problem*}{Problem}

\newtheorem*{mdresult}{Main Result}
\newenvironment{result}{\begin{mdframed}[backgroundcolor=lightgray!40,topline=false,rightline=false,leftline=false,bottomline=false,innertopmargin=2pt]\begin{mdresult}}{\end{mdresult}\end{mdframed}}

\allowdisplaybreaks

\DeclareMathOperator*{\argmax}{arg\,max}

\newcommand{\mi}[2]{\ensuremath{\II(#1 \,; #2)}}

\renewcommand{\qed}{\nobreak \ifvmode \relax \else
      \ifdim\lastskip<1.5em \hskip-\lastskip
      \hskip1.5em plus0em minus0.5em \fi \nobreak
      \vrule height0.75em width0.5em depth0.25em\fi}

\usepackage[T1]{fontenc}
\usepackage[utf8]{inputenc}

\newcommand{\ourinfo}{Supported in part by National Science Foundation grants CCF-1552909, CCF-1617851, and IIS-1447470.}

\input{macros}

\title{Combinatorial Auctions Do Need Modest Interaction}
\author{Sepehr Assadi\thanks{\ourinfo}\\University of Pennsylvania \\ {sassadi@cis.upenn.edu}}

\date{}

\begin{document}
\maketitle

\thispagestyle{empty}
\input{abstract}
\clearpage
\setcounter{page}{1}

\input{intro}

\input{prelim}

\input{sim}

\input{multi-round}

\input{conclusion}

\subsection*{Acknowledgements}
I thank my advisor Sanjeev Khanna for many helpful advice and comments. I am also grateful to Jamie Morgenstern for helpful discussions in the earlier stages of this work and to Matthew Weinberg for 
bringing~\cite{BravermanMW17} to my attention. Finally, I would like to thank the anonymous reviewers of EC 2017 for many insightful comments and suggestions

\bibliographystyle{abbrv}
\bibliography{general}

\end{document}

%% file: macros.tex
\newcommand{\toShrink}{-.20cm}
\newcommand{\toShrinkEnu}{-.2cm}

\newenvironment{tbox}{\begin{tcolorbox}[
		enlarge top by=5pt,
		enlarge bottom by=5pt,
		 boxsep=0pt,
                  left=4pt,
                  right=4pt,
                  top=10pt,
                  arc=0pt,
                  boxrule=1pt,toprule=1pt,
                  colback=white
                  ]
	}
{\end{tcolorbox}}

\newcommand{\textbox}[2]{
{
\begin{tbox}
\textbf{#1}
{#2}
\end{tbox}
}
}

\newcommand{\itfacts}[1]{Fact~\ref{fact:it-facts}-(\ref{part:#1})\xspace}

\newcommand{\algline}{
  \rule{0.5\linewidth}{.1pt}\hspace{\fill}%
  \par\nointerlineskip \vspace{.1pt}
}

\newcommand{\tvd}[2]{\ensuremath{\norm{#1 - #2}_{tvd}}}
\newcommand{\Ot}{\ensuremath{\widetilde{O}}}
\newcommand{\eps}{\ensuremath{\varepsilon}}

\newcommand{\Bracket}[1]{\Big[#1\Big]}
\newcommand{\bracket}[1]{\left[#1\right]}
\newcommand{\paren}[1]{\ensuremath{\left(#1\right)}\xspace}
\newcommand{\card}[1]{\left\vert{#1}\right\vert}

\newcommand{\IR}{\ensuremath{\mathbb{R}}}

\newcommand{\norm}[1]{\ensuremath{\|#1\|}}

\newcommand{\set}[1]{\ensuremath{\left\{ #1 \right\}}}
\newcommand{\poly}{\mbox{\rm poly}}
\newcommand{\polylog}{\mbox{\rm  polylog}}

\DeclareMathOperator*{\Exp}{\ensuremath{{\mathbb{E}}}}
\DeclareMathOperator*{\Prob}{\ensuremath{\textnormal{Pr}}}
\renewcommand{\Pr}{\Prob}

\newcommand{\Ex}{\Exp}

\newcommand{\etal}{et al.\xspace}

\newcommand{\dist}{\ensuremath{\mathcal{D}}}

\newcommand{\distr}[1]{\ensuremath{\dist_{#1}}}

\newcommand{\FS}{\ensuremath{\mathcal{S}}}
\newcommand{\FF}{\ensuremath{\mathcal{F}}}

\newcommand{\jstar}{\ensuremath{j^{\star}}}

\newcommand{\Ins}{\ensuremath{\mathcal{I}}}

\newcommand{\Prot}{\ensuremath{\Pi}}
\newcommand{\prot}{\ensuremath{\pi}}

\newcommand{\II}{\ensuremath{\mathbb{I}}}
\newcommand{\HH}{\ensuremath{\mathbb{H}}}

\newcommand{\Istar}{\ensuremath{I^{\star}}}

\newcommand{\FN}{\ensuremath{\mathcal{N}}}

\newcommand{\thetas}{\ensuremath{\theta^{\star}}}

\renewcommand{\neg}[1]{\ensuremath{\overline{#1}}}

\newcommand{\event}{\ensuremath{\mathcal{E}}}

\newcommand{\FA}{\ensuremath{\mathcal{A}}}

\renewcommand{\tvd}[2]{\ensuremath{\norm{#1-#2}}}

\newcommand{\errs}{\ensuremath{\textnormal{errs}}\xspace}

\newcommand{\DD}[2]{\ensuremath{\mathbb{D}(#1~||~#2)}}

\newcommand{\distribution}[1]{\ensuremath{\textnormal{dist}(#1)}}

\newcommand{\supp}[1]{\ensuremath{\textnormal{supp}(#1)}}

\newcommand{\sw}{\ensuremath{\textnormal{\textsf{sw}}}}

\newcommand{\bid}{\ensuremath{i}}

%% file: abstract.tex
\begin{abstract}
	We study the necessity of interaction for obtaining efficient allocations in combinatorial auctions with \emph{subadditive bidders}.  
	This problem was originally introduced by Dobzinski, Nisan, and Oren~\cite{DobzinskiNO14} as the following simple market scenario: $m$ items are to be allocated 
	among $n$ bidders in a distributed setting where bidders valuations are private and hence \emph{communication} is needed to obtain an efficient allocation. 
	The communication happens in rounds: in each round, each bidder, \emph{simultaneously} with others, broadcasts a message to all parties involved. 
	At the end, the central planner computes an allocation solely based on the communicated messages. Dobzinski~\etal~\cite{DobzinskiNO14} showed
	that (at least some) interaction is necessary for obtaining any efficient allocation: no non-interactive ($1$-round) protocol with polynomial communication (in the number of items and bidders) can 
	achieve approximation ratio better than $\Omega(m^{{1}/{4}})$, while for any $r \geq 1$, there exists $r$-round protocols  that achieve $\Ot(r \cdot m^{{1}/{r+1}})$ approximation with polynomial communication; 
	in particular, $O(\log{m})$ rounds of interaction \emph{suffice} to obtain an (almost) efficient allocation, i.e., a \polylog{(m)}-approximation. 
	
	A natural question at this point is to identify the ``right'' level of interaction (i.e., number of rounds) \emph{necessary} to obtain an efficient allocation. 
	In this paper, we resolve this question by providing an almost tight \emph{round-approximation} tradeoff for this problem: we show that for any $r \geq 1$, 
	any $r$-round protocol that uses $\poly(m,n)$ bits of communication can only approximate the social welfare up to a factor of $\Omega(\frac{1}{r} \cdot m^{{1}/{2r+1}})$.  
	This in particular implies that $\Omega(\frac{\log{m}}{\log\log{m}})$ rounds of interaction are \emph{necessary} for obtaining any efficient allocation (i.e., a constant or even a $\polylog{(m)}$-approximation) in these markets. 
	Our work builds on the recent multi-party round-elimination technique of Alon, Nisan, Raz, and Weinstein~\cite{AlonNRW15} -- used to prove similar-in-spirit lower bounds for round-approximation tradeoff
	in unit-demand (matching) markets -- and settles an open question posed by Dobzinski~\etal~\cite{DobzinskiNO14} and Alon~\etal~\cite{AlonNRW15}.
\end{abstract}

%% file: intro.tex
\section{Introduction}\label{sec:intro}

In a combinatorial auction, $m$ items in $M$ are to be allocated between $n$ bidders (or players\footnote{Throughout the paper, we use the terms ``bidder'' and ``player'' interchangeably.}) in $N$ with valuation functions $v_i : 2^{M} \rightarrow \IR_{+}$. The goal is to find a collection of disjoint {bundles} $A_1,\ldots,A_n$ of items in $M$ (an \emph{allocation}), that maximizes \emph{social welfare} defined as the sum of bidder's valuations 
for the allocated bundles, i.e., $\sum_{i \in N} v_i(A_i)$. This paper studies the tradeoff between the amount of interaction between the bidders and the efficiency of the allocation in combinatorial auctions. 

In our model, each bidder $i \in N$ only knows the valuation function $v_i$ and hence the bidders need to communicate to obtain an efficient allocation.
Communication happens in rounds. In each round, each bidder $i$, \emph{simultaneously} with others,
broadcasts a message to all parties involved, based on the valuation function $v_i$ and messages in previous rounds. 
In the last round, the central planner outputs the allocation solely based on the communicated messages. 
Notice that a ``trivial solution'' in this setting is for all players to communicate their entire input to the central planner who can then compute an efficient allocation; 
however, such a protocol is clearly infeasible in most settings as it has an enormous communication cost. As such, we are interested in protocols with significantly less communication cost,
typically \emph{exponentially smaller} than the input size.

This model was first introduced by Dobzinski, Nisan, and Oren~\cite{DobzinskiNO14} to address the following fundamental question in economics: \emph{``To what extent is interaction between individuals
required in order to efficiently allocate resources between themselves?''}. They considered this problem for two different classes of valuation functions:  \emph{unit-demand} valuations
and \emph{subadditive} valuations (see Section~\ref{sec:auctions}). 
For both settings, they showed that (at least some) interaction is necessary to obtain an efficient allocation: non-interactive (aka $1$-round or simultaneous) protocols have enormous communication cost  
compared to interactive ones, while even allowing a modest amount of interaction allows for finding an (approximately) efficient allocation. We now elaborate more on these results. 

For the case of matching markets with $n$ unit-demand bidders and $n$ items (and hence input-size of $n$ bits per each player), Dobzinski~\etal~\cite{DobzinskiNO14} proved a
lower bound of $\Omega(\sqrt{n})$ on the approximation ratio of any simultaneous protocol that communicates $n^{o(1)}$ bits per each bidder. On the other hand, they
showed that for any $r \geq 1$, there exists an $r$-round protocol that achieves an $O(n^{1/r+1})$ approximation by sending $O(\log{n})$ bits per each bidder in each round. 
For the more general setting of combinatorial auctions with $n$ subadditive bidders and $m$ items (and hence input-size of $\exp(m)$ bits per each player), 
they showed that the best approximation ratio achievable by simultaneous protocols with $\poly(m,n)$ communication is $\Omega(m^{1/4})$, while for any $r \geq 1$, there exists 
$r$-round protocols that achieve an approximation ratio of $\Ot(r \cdot m^{1/r+1})$. These results imply that in such markets, logarithmic rounds of interaction in the market size 
\emph{suffice} to obtain an (almost) efficient allocation, i.e., a $\polylog{(m)}$-approximation. 

A natural question left open by~\cite{DobzinskiNO14} was to identify the amount of interaction \emph{necessary} to obtain an efficient allocation in these markets. 
Recently, Alon, Nisan, Raz, and Weinstein~\cite{AlonNRW15} provided a partial answer to this question for matching markets: for any $r \geq 1$, any $r$-round protocol 
for unit-demand bidders in which each bidder sends at most $n^{o(1)}$ bits in each round can only achieve an $\Omega(n^{1/5^{r+1}})$ approximation~\cite{AlonNRW15}. This implies
that at least $\Omega(\log\log{n})$ rounds of interaction is necessary to achieve an efficient allocation in matching markets. Alon~\etal~\cite{AlonNRW15} further conjectured that the 
``correct'' lower bound for the convergence rate in this setting is $\Omega(\log{n})$; in other words, $\Omega(\log{n})$ rounds of interaction are necessary for achieving an efficient allocation. 

Despite this progress for matching markets, the best known lower bounds for the more general setup of combinatorial auctions with subadditive bidders remained the aforementioned $1$-round lower bound of~\cite{DobzinskiNO14}, 
and a $(2-\eps)$-approximation (for every constant $\eps > 0$) for any polynomial communication protocol with unrestricted number of rounds~\cite{DobzinskiNS05}. Indeed, obtaining better lower bounds
for $r$-round protocols was posed as an open problem by Alon~\etal~\cite{AlonNRW15} who also mentioned that: ``from a communication complexity perspective, lower bounds
in this setup are more compelling, since player valuations require exponentially many bits to encode, hence interaction has
the potential to reduce the overall communication from exponential to polynomial.''.

\subsection{Our Results and Techniques} \label{sec:result}

In this paper, we resolve the aforementioned open question of Dobzinski~\etal~\cite{DobzinskiNO14} and Alon~\etal~\cite{AlonNRW15} by proving an almost tight \emph{round-approximation} tradeoff for polynomial communication
protocols in subadditive combinatorial auctions. 

\vspace{5pt}
\begin{result}
	For any $r \geq 1$, any $r$-round protocol (deterministic or randomized) for combinatorial auctions with \emph{subadditive} bidders that uses polynomial communication
	can only achieve an approximation ratio of $\Omega(\frac{1}{r} \cdot m^{{1}/{\Theta(r)}})$ to the social welfare. 
\end{result}
\vspace{5pt}

We remark that this lower bound holds even when the bidders valuations are XOS functions, a strict subclass of subadditive valuations (see Section~\ref{sec:auctions} for definition). 

Our main result, combined with the upper bound result of~\cite{DobzinskiNO14}, provides a near-complete understanding of the power of each additional round in improving the quality of 
the allocation in subadditive combinatorial auctions. Moreover, an immediate corollary of our result is that in these markets, $\Omega(\frac{\log{m}}{\log\log{m}})$ rounds of interaction
are \emph{necessary} to achieve any efficient allocation (i.e., constant or polylogarithmic approximation), which is \emph{tight} up to an $O(\log\log{m})$ factor. 
The qualitative message of this theoretical result is clear: \emph{a modest amount of interaction between individuals in a market is crucial for obtaining an efficient allocation}.

Our first step in establishing this result is proving a new lower bound for simultaneous ($1$-round) protocols. We deviate 
from~\cite{DobzinskiNO14} by considering the problem of estimating the \emph{value} of social welfare as opposed to finding the actual allocation; this 
problem can only be harder in terms of proving a lower bound as any protocol that can find an approximate allocation can also be used to estimate
the value of social welfare with one additional round and $O(n)$ additional communication using a trivial reduction (see Section~\ref{sec:communication}). As a result, the
kind of combinatorial arguments used in~\cite{DobzinskiNO14} seem not 
sufficient for our purpose and we instead prove our lower bound using information-theoretic machinery and in particular a direct-sum style argument. This counterintuitive switch to establishing a lower bound for a seemingly harder 
problem however leads to a more modular proof that allows us to further carry out our results to multi-round protocols. 

We establish our multi-round lower bound following the multi-party round-elimination technique of~Alon~\etal~\cite{AlonNRW15}. We create a recursive family of hard distributions 
$\dist_1,\dist_2,\ldots$ whereby for any $r \geq 1$, $\dist_r$ is the hard input distribution for $r$-round protocols. Each instance in $\dist_r$ is a careful combination of \emph{exponentially many} sub-instances 
sampled from $\dist_{r-1}$. One of these sub-instances is ``special'' in that to solve the original instance, the players also need to solve this special sub-instance completely. 
On the other hand, the players are not able to identify this special sub-instance \emph{locally} and hence need to spend one round of interaction \emph{only} for this purpose. 
In other words, we prove that the first round of protocol does not convey much information about the special instance beyond its identity. Using a further round-elimination argument, we inductively show that since solving the special instance is hard for $(r-1)$-round protocols, solving the original 
instance should be hard for $r$-round protocols as well. 

Similar to~\cite{AlonNRW15}, and unlike typical two-player round-elimination arguments (see, e.g.~\cite{MiltersenNSW98,SenV08}), eliminating a round in our round-elimination argument requires a reduction from 
``low dimensional'' instances (with fewer players and items) to ``high dimensional'' instances. This reduction is delicate as the players need to ``complete'' their inputs in the higher dimensional instance by \emph{independently} sampling the ``missing part'' \emph{conditioned on the first message} of the protocol \emph{without any further communication}, while this distribution is a \emph{correlated} distribution. 

Furthermore, in contrast to~\cite{AlonNRW15}, our sub-instances in each distribution are \emph{overlapping} (as otherwise exponentially many sub-instances cannot be embedded inside a single polynomially larger instance) and hence may interfere with each other, potentially diminishing the role of the special instance. We overcome this obstacle by embedding these sub-instances based on a family
of \emph{small-intersection sets} to limit the potential overlap between the sub-instances and prove that solving the
special instance is crucial even in the presence of these overlaps. It is worth pointing out that this approach allows us to avoid the doubly-exponential 
rate of growth in the size of instances across different rounds in~\cite{AlonNRW15}, resulting in exponentially better dependence on the parameter $r$ in our lower bound compared to~\cite{AlonNRW15}. 
Finally, since our lower bound is for estimating the \emph{value} of social welfare (as opposed to finding an allocation),
we need a different embedding argument in our reduction than the one used in~\cite{AlonNRW15}\footnote{Our problem corresponds to the problem of estimating the \emph{size} of a maximum matching as opposed to 
finding an \emph{approximate} matching in the setting of~\cite{AlonNRW15}. To the best of our knowledge, no non-trivial lower bounds are known for the matching size estimation problem in the setting of~\cite{AlonNRW15}; 
see~\cite{AssadiKL17} for further details.}. In particular, we now embed the low dimensional instance in \emph{multiple} places of the
high dimensional instance as opposed to only one.

\subsection{Other Related Work}\label{sec:related}
Communication complexity of combinatorial auctions has received quite a lot of attention in the literature. It is known that for \emph{arbitrary valuations}, exponential 
amount of communication is needed to obtain an $\paren{m^{1/2-\eps}}$-approximate allocation (for every constant $\eps > 0$)~\cite{NisanS06} (see also~\cite{Nisan02}), and this
is also tight~\cite{AharoniEL88,Lovasz75,BriestKV11,LaviS11}. 
For \emph{subadditive valuations}, a constant factor approximation to the social welfare
can be achieved in our model using only polynomial communication~\cite{DobzinskiNS05,DobzinskiS06,Feige09,FeigeV06,LehmannLN06,Vondrak08,DuttingK17} 
(and polynomially many rounds of interaction); in particular, Feige~\cite{Feige09} developed a $2$-approximation polynomial communication protocol for this problem and Dobzinski, Nisan, and Schapira~\cite{DobzinskiNS05} proved 
that obtaining  $(2-\eps)$-approximation (for any constant $\eps > 0$) requires exponential communication (regardless of the number of rounds). Moreover, Dutting and Kesselheim~\cite{DuttingK17} designed an 
$O(\log{m})$-approximation protocol with polynomial communication for subadditive combinatorial auctions in which each bidder needs to communicate exactly once; however, this protocol still requires $n$ rounds of interaction in our 
model as the players need to communicate in a round-robin fashion making the message sent by a bidder crucially depending on the messages communicated earlier by the previous bidders.

Another line of relevant research considers the case where the valuation of the bidders are chosen \emph{independently} from a commonly known distribution (see, e.g.~\cite{FeldmanFGL13,FeldmanGL15}) and aims
to design ``simple'' and simultaneous protocols that achieve an efficient allocation. The main difference between this setting and ours is that we are interested in arbitrary distributions of inputs for the bidders which are not necessarily 
product distributions; as already shown by the strong impossibility results of~\cite{DobzinskiNO14}, the aforementioned type of protocols cannot provably exist in our model when input distributions are correalted. Finally, 
we point out that ``incompressability'' results are also known for subadditive valuations: any polynomial-length encoding of subadditive
valuations must lose $\Omega(\sqrt{m})$ in precision~\cite{BadanidiyuruDFKNR12,BalcanCIW12}. 

We refer the interested reader to~\cite{DobzinskiNO14} for a comprehensive summary of related work and further discussion on the role of interaction in markets.


%% file: prelim.tex
\section{Preliminaries}\label{SEC:PRELIM}

\paragraph{Notation.} For any integer $a \geq 1$, we let $[a]:=\set{1,\ldots,a}$. We say that a set $S \subseteq [n]$ with $\card{S}=s$ is a \emph{$s$-subset} of $[n]$. 
For a $k$-dimensional tuple $X = (X_1,\ldots,X_k)$ and index $i \in [k]$, we define $X^{<i}:= (X_1,\ldots,X_{i-1})$ and $X^{-i}:=(X_1,\ldots,X_{i-1},X_{i+1},\ldots,X_k)$.
We use capital letters to denote random variables. For a random variable $A$, $\supp{A}$ denotes the support of $A$ and $\distribution{A}$ denotes its distribution. We further define $\card{A} := \log{\card{\supp{A}}}$. 
We write $A \perp B \mid C$ to denote that $A$ and $B$ are independent conditioned on $C$. We use ``w.p.'' to mean ``with probability''. 

\paragraph{Concentration bounds.} Throughout, we use the following version of Chernoff bound for negatively correlated random variables first proved by~\cite{PanconesiS97}; see, also~\cite{ConcentrationBook,ImpagliazzoK10}. 

\begin{proposition}[Chernoff bound]\label{prop:chernoff}
	Let $X_1,\ldots,X_n$ be \emph{negatively correlated} random variables taking values in $[0,1]$ and let $X:= \sum_{i=1}^{n} X_i$. Then, for any $\alpha \geq 2e^2$,
	\[ \Pr\paren{X \geq \alpha \cdot \Ex\bracket{X}} \leq \exp\paren{-\Omega(\alpha \cdot \Ex\bracket{X})} \]
\end{proposition}

\paragraph{Intersecting families.} The following combinatorial construction plays a crucial role in our proofs. 

\begin{definition}
	A \emph{$(p,q,t,\ell)$-intersecting family} $\FF$ is a collection of $p$ subsets of $[q]$ each of size $t$, such that for any two distinct sets $S,T \in \FF$, $\card{S \cap T} \leq \ell$. 
\end{definition}

We prove the existence of an exponentially large intersecting family with a small pair-wise intersection, using a probabilistic argument. 

\begin{lemma}\label{lem:intersecting-families}
	For any integer $r \geq 1$, any parameter $\eps > 0$, and any integer $k \geq \paren{2e^2 \cdot r^2}^{\frac{1}{\eps}}$,  there
	exists a $(p,q,t,\ell)$-intersecting family with $p = \exp\paren{\Theta(k^{2r-2+\eps})}$, $q = k^{2r} + r \cdot k^{2r-1}$, $t = r \cdot k^{2r-1}$, and $\ell = k^{2r-2 + \eps}$. 
\end{lemma}
\begin{proof}
	Let $\FF$ be a family of $p$ sets (for $p$ to be determined later), each chosen independently and uniformly at random from all $t$-subsets of $[q]$. Fix any pair of sets $S,T \in \FF$; 
	for each element $a \in S$, define the random variable $X_a \in \set{0,1}$ which is $1$ iff $a \in T$ also. We have $\Ex\bracket{X_a} \leq {r}/{k}$.  Let $X = \sum_{a \in S} X_a$ denotes 
	$\card{S \cap T}$; hence $\Ex\bracket{X} = r^2 \cdot k^{2r-2}$. Since $X_a$'s are negatively correlated random variables, by Chernoff bound (Proposition~\ref{prop:chernoff} with $\alpha = k^{\eps}/r^2 \geq 2e^2$ by 
	lemma statement), 
	\begin{align*}
		\Pr\paren{\card{S \cap T} > \ell} = \Pr\paren{X > k^{2r-2+\eps}} = \Pr\paren{X > k^{\eps}/r^2 \cdot \Ex\bracket{X}} \leq \exp\paren{-\Omega(k^{2r-2+\eps})} 
	\end{align*}
	By a union bound over all possible choices for $S,T \in \FF$, 
	\begin{align*}
		\Pr\paren{\exists~S,T \in \FF: \card{S \cap T} > \ell} \leq \sum_{S \neq T \in \FF} \Pr\paren{\card{S \cap T} > \ell} \leq {{p} \choose{2}} \cdot \exp\paren{-\Omega(k^{2r-2+\eps})} 
	\end{align*}
	Taking $p = \exp\paren{\Theta(k^{2r-2+\eps})}$ ensures that with some non-zero probability, the set $\FF$ is a $(p,q,t,\ell)$-intersecting family, implying the existence of such a family. 
\end{proof}

\subsection{Combinatorial Auctions} \label{sec:auctions}
We have a set $N$ of $n$ bidders, and a set $M$ of $m$ items. Each bidder $\bid \in N$ has a valuation function $v_\bid: 2^{M} \rightarrow \IR_{+}$, which assigns 
a value to each \emph{bundle} of items (we assume $v_\bid(\emptyset) = 0$ and $v_\bid(\cdot)$ is non-decreasing). The goal is to maximize the \emph{social welfare} defined as $\max_{(A_1,\ldots,A_n)} \sum_{\bid \in N} v_\bid(A_\bid)$, where $(A_1,\ldots,A_n)$ ranges over all possible allocations of items in $M$ to bidders in $N$ such that bidder $\bid$ receives the bundle $A_\bid$.  

A valuation function $v(\cdot)$ is considered \emph{subadditive} iff for any two bundles of items $S,T \subseteq M$, $v(S \cup T) \leq v(S) + v(T)$. A valuation function is 
\emph{additive} iff for any bundle $S \subseteq M$, $v(S) = \sum_{j \in S}v(\set{j})$. A valuation function is $\emph{XOS}$ iff there exists $r$ additive valuation functions $a_1,\ldots,a_r$ 
such that for all bundles $S \subseteq M$, $v(S) = \max_{r} a_r(S)$. Each function $a_j$ is called a \emph{clause} of $v$ and for any bundle $S$, any clause $a \in \argmax_r a_r(S)$ is referred 
to as a \emph{maximizing clause} of $S$. Finally, a valuation function is \emph{unit-demand} iff for any $S \subseteq M$, $v(S) = \max_{j \in S} v(\set{j})$. It is easy to verify that both XOS and unit-demand functions are also 
subadditive. 

Notice that in general, subadditive and XOS valuation functions require $\exp(m)$ many bits for representation, while unit-demand valuation functions can be
 represented with $O(m)$ numbers, i.e., by describing the value of each singleton set. As such, in subadditive combinatorial auctions, we are interested in protocols that can reduce the communication from exponential in $m$ to polynomial, while in unit-demand auctions, we mainly seek protocols that
reduce the communication from linear in $m$ to logarithmic.

\subsection{Communication Model} \label{sec:communication}

We use the ({number-in-hand}) multiparty communication model with {shared blackboard}: there are $n$ players (corresponding to the bidders) receiving inputs
$(x_1,\ldots,x_n)$, jointly distributed according to a prior distribution $\dist$ on $\mathcal{X}_1 \times \ldots \times \mathcal{X}_n$.  The communication proceeds in \emph{rounds} whereby in each round $r$, 
the players \emph{simultaneously} write a message on a \emph{shared blackboard} visible to all parties. In a deterministic protocol, the message sent by any player $\bid$ in each round can only depend on the private 
input of the player, i.e., $x_\bid$, plus the messages of all players in previous rounds, i.e., the content of the blackboard. In a randomized protocol, we further allow the players to have access to both public and private randomness
and the message of players can depend on them as well. 

For a protocol $\prot$, we use $\Prot=(\Prot_1,\ldots,\Prot_n)$ to denote the transcript of the message communicated by the players (i.e., the content of the blackboard). In addition to the $n$ players, there exists also a $(n+1)$-th 
party called the \emph{referee} which does not have any input, and is responsible for outputting the answer in the last round, solely based on content of the blackboard $\Prot$ (plus the public randomness in case of randomized 
protocols). Finally, the \emph{communication cost} of the protocol $\prot$, denoted by $\norm{\prot}$, is the sum of worst-case length of the messages communicated by all players, i.e., $\norm{\prot} = \sum_{i=1}^{n} \card{\Prot_i}$. 

\paragraph{Approximation guarantee.} We consider protocols that are required to estimate the \emph{maximum value of social welfare} in any instance $I$ of a combinatorial auction (denoted by $\sw(I)$).
More formally, a $\delta$-error $\alpha$-approximation protocol needs to, for each input instance $I$ sampled from $\dist$, output a number in the 
range $[\frac{1}{\alpha} \cdot \sw(I) , \sw(I)]$ w.p. at least $1-\delta$, where the randomness is over the distribution $\dist$ (and the randomness of protocol in case of randomized protocols). 

This problem is provably easier than finding an approximate allocation in the interactive setting: any $r$-round protocol for finding an approximate allocation can be used to obtain an $(r+1)$-round protocol for estimating the 
value of social welfare with $O(n)$ additional communication; simply compute the approximate allocation in the first $r$ rounds and spend one additional round in which each player declares her value for the assigned bundle
to the referee. It was shown very recently in~\cite{BravermanMW17} that this \emph{loss of one round} in the reduction is unavoidable (see Section~\ref{SEC:SIM} for further details). However, 
this extra one round is essentially negligible for our purpose as we are interested in the asymptotic dependence of the approximation ratio and the number of rounds.

\subsection{Tools from Information Theory}\label{sec:info}

\input{info}

%% file: info.tex
We briefly review some basic definitions and facts from information theory that are used in this paper. 
We refer the interested reader to the excellent text by Cover and Thomas~\cite{ITbook} for an introduction to this field,
and the proofs of the claims in this section (see Chapter 2). 

In the following, we denote the \emph{Shannon Entropy} of a random variable $A$ by
$\HH(A)$ and the \emph{mutual information} of two random variables $ A$ and $ B$ by
$\mi{A}{B} = \HH( A) - \HH( A \mid  B) = \HH( B) - \HH( B \mid  A)$. 
We use $H_2(\cdot)$ to denote the binary entropy function where for any real number $0
< \delta < 1$, $H_2(\delta) := \delta\log{\frac{1}{\delta}} +
(1-\delta)\log{\frac{1}{1-\delta}}$. 
The proof of the following basic properties of entropy and mutual information can be found in~\cite{ITbook}, Chapter~2.

\begin{fact}\label{fact:it-facts}
  Let $ A$, $ B$, and $ C$ be three (possibly correlated) random variables.
   \begin{enumerate}
  \item \label{part:uniform} $0 \leq \HH( A) \leq \card{ A}$, and $\HH( A) = \card{ A}$
    iff $ A$ is uniformly distributed over its support.
  \item \label{part:info-zero} $\mi{A}{B \mid C} \geq 0$. The equality holds iff $ A$ and
    $ B$ are \emph{independent} conditioned on $C$.
  \item \label{part:cond-reduce} $\HH( A \mid  B, C) \leq \HH( A \mid  B)$.  The equality holds iff $ A \perp C \mid B$.
  \item \label{part:chain-rule} $\mi{A, B}{C} = \mi{A}{C} + \mi{B}{C \mid  A}$ (\emph{chain rule of mutual information}).
  \item \label{part:data-processing} Suppose $f(A)$ is a deterministic function of $A$, then $\mi{f(A)}{B \mid C} \leq \mi{A}{B \mid C}$ (\emph{data processing inequality}). 
   \end{enumerate}
\end{fact}

The following Fano's inequality states that if a random variable $A$ can be used to estimate the value of another random variable $B$, 
then $A$ should ``consume'' most of $B$'s entropy.

\begin{fact}\label{fact:fano}
	Let $A,B$ be random variables and $f$ be a function that given $A$ predicts a value for $B$. Suppose $B$ is binary and 
	$\Pr\paren{f(A) \neq B} \leq \delta$, then $\HH(B \mid A) \leq H_2(\delta)$. 
\end{fact}

We also use the following two simple propositions.

\begin{proposition}\label{prop:info-increase}
  For variables $ A,  B,  C, D$, if $ A \perp D \mid C$, then, $\mi{A}{B \mid C} \leq \mi{A}{B \mid  C,  D}$.
\end{proposition}
 \begin{proof}
  Since $ A$ and $ D$ are independent conditioned on $C$, by
  \itfacts{cond-reduce}, $\HH( A \mid  C) = \HH( A \mid
   C, D)$ and $\HH( A \mid  C, B) \ge \HH( A \mid  C, B, D)$.  We have,
	 \begin{align*}
	  \mi{A}{B \mid  C} &= \HH( A \mid  C) - \HH( A \mid  C, B) = \HH( A \mid  C, D) - \HH( A \mid  C, B) \\
	  &\leq \HH( A \mid  C, D) - \HH( A \mid  C, B, D) = \mi{A}{B \mid  C, D}
	\end{align*}
\end{proof}

\begin{proposition}\label{prop:info-decrease}
  For variables $ A,  B,  C,D$, if $ A \perp D \mid B,C$, then, $\mi{A}{B \mid  C} \geq \mi{A}{B \mid  C,  D}$.
\end{proposition}
 \begin{proof}
 Since $A \perp D \mid B,C$, by \itfacts{cond-reduce}, $\HH(A \mid B,C) = \HH(A \mid B,C,D)$. Moreover, since conditioning can only reduce the entropy (again by \itfacts{cond-reduce}), 
  \begin{align*}
 	\mi{A}{B \mid  C} &= \HH(A \mid C) - \HH(A \mid B,C) \geq \HH(A \mid D,C) - \HH(A \mid B,C) \\
	&= \HH(A \mid D,C) - \HH(A \mid B,C,D) = \mi{A}{B \mid C,D}
 \end{align*}
\end{proof}

For two distributions $\mu$ and $\nu$ over the same probability space, the \emph{Kullback-Leibler divergence} between $\mu$ and $\nu$ is defined as $\DD{\mu}{\nu}:= \Ex_{a \sim \mu}\Bracket{\log\frac{\Pr_\mu(a)}{\Pr_{\nu}(a)}}$.
We have,
\begin{fact}\label{fact:kl-info}
	For random variables $A,B,C$, 
	\[\mi{A}{B \mid C} = \Ex_{(b,c) \sim \distribution{B,C}}\Bracket{ \DD{\distribution{A | C=c}}{\distribution{A |B=b,C=c}}}.\] 
\end{fact}

We denote the \emph{total variation distance} between two distributions $\mu$ and $\nu$ over the same probability space $\Omega$ by $\tvd{\mu}{\nu} = \frac{1}{2} \cdot \sum_{x \in \Omega} \card{\Pr_{\mu}(x) - \Pr_{\nu}(x)}$. 

The following Pinskers' inequality bounds the total variation distance between two distributions based on their KL-divergence, 

\begin{fact}[Pinsker's inequality] \label{fact:pinskers}
	For any two distributions $\mu$ and $\nu$, $\tvd{\mu}{\nu} \leq \sqrt{\frac{1}{2} \cdot \DD{\mu}{\nu}}$. 
\end{fact}

\begin{fact}\label{fact:tvd-small}
	Suppose $\mu$ and $\nu$ are two distributions for an event $\event$, then, $\Pr_{\mu}(\event) \leq \Pr_{\nu}(\event) + \tvd{\mu}{\nu}$. 
\end{fact}

%% file: sim.tex
\section{Warm Up: A Lower Bound for Simultaneous Protocols}\label{SEC:SIM}

Our main lower bound result is based on analyzing a recursive family of distributions. As a warm up,
we analyze the base case of this recursive construction in this section and prove a lower bound for $1$-round (simultaneous) protocols. Formally, 

\begin{theorem}\label{thm:sim-lower}
	For any sufficiently small constant $\eps > 0$, any simultaneous protocol (possibly randomized) for combinatorial auctions with subadditive (even XOS) bidders 
	that can approximate the value of social welfare to a factor of $m^{\frac{1}{3} - \eps}$ requires $\exp\paren{m^{\Omega(\eps)}}$ bits of communication. 
\end{theorem}

It is worth mentioning that the bound established in Theorem~\ref{thm:sim-lower} on the approximation ratio of simultaneous protocols is \emph{tight}. Previously, Dobzinski~\etal~\cite{DobzinskiNO14} developed a
simultaneous protocol that can approximate the social welfare up to an $\Ot(m^{1/3})$ factor using only polynomial communication. As such, Theorem~\ref{thm:sim-lower} 
already makes a small contribution for simultaneous protocols. But more importantly, this theorem sets the stage for our main lower bound result in Section~\ref{SEC:MULTI}. 

As pointed out earlier, Dobzinski~\etal~\cite{DobzinskiNO14} have previously proved an $\Omega(m^{1/4})$ lower bound on the approximation ratio of the protocols that can 
find an approximate allocation. We should remark that this lower bound of~\cite{DobzinskiNO14} and our lower bound in Theorem~\ref{thm:sim-lower} are \emph{incomparable} in that neither imply (or strengthen) the other. 
The reason is that while the problem of estimating the social welfare is provably easier than the problem of finding an approximate allocation, the reduction requires one additional round of interaction and hence, in general, a 
simultaneous protocol for the problem of finding the allocation only implies a $2$-round (and not a simultaneous) protocol for the social welfare estimation problem\footnote{Note however that the $\Ot(m^{1/3})$-approximation protocol 
of~\cite{DobzinskiNO14} can already compute the welfare of the allocated allocation and hence does \emph{not} need an additional round for estimating the social welfare, implying the tightness of the bounds in Theorem~
\ref{thm:sim-lower}.}. Interestingly, for the case of $n=2$ players, Braverman~\etal~\cite{BravermanMW17} very recently showed that the problem of estimating the social welfare is indeed provably \emph{harder} than finding an approximate allocation for \emph{simultaneous} protocols. In the light of this
result, it seems plausible that one can indeed improve the protocol of~\cite{DobzinskiNO14} and find an $O(m^{1/4})$-approximation protocol for finding an approximate allocation (matching the lower bound of~\cite{DobzinskiNO14}); 
however, Theorem~\ref{thm:sim-lower} suggests that if such a protocol exists, it necessarily should be oblivious to the welfare of the allocation it provides.

\subsection{A Hard Input Distribution for Simultaneous Protocols}\label{sec:sim-dist}

In this section, we propose a hard input distribution $\dist_1$ for simultaneous protocols and state several of its properties that are needed in proving the lower bound
for this distribution. We start by providing an informal description of the distribution $\dist_1$. 

Let $k$ be an integer and consider a set $N$ of $n=k^2$ players and $M$ of $m = k^3$ items. Each bidder $i \in N$, is given an exponentially large (in $k$) collection $\FF_\bid$ of item-sets of size $k$ each, such that 
for all $S \subseteq M$, $v_\bid(S) = \max_{T \in \FF_\bid} \card{S \cap T}$ (recall that the input to player $i$ is the valuation function $v_\bid(\cdot)$). Additionally, the sets in $\FF_\bid$ are ``barely overlapping'', in the sense that
for any two sets $S,T \in \FF_\bid$, $\card{S \cap T} < k^{\eps}$ (for any constant $\eps > 0$). 

This construction ensures that \emph{locally} each player is confronted with exponentially many high value bundles (sets in $\FF_\bid$) that look ``exactly the same''. However, these collections across different players
are chosen in a correlated way such that except for a single ``special bundle'' $T_j \in \FF_\bid$ (for each $\bid \in N$), the items in all other bundles are chosen (mostly) from a (relatively small) set of $k^2$ ``shared'' items across all players. The special bundles on the other hand consist of ``unique'' items. This imply that \emph{globally} each player is assigned a special bundle and these special bundles are crucial to 
obtaining any $k^{1-\eps}$-approximate allocation (recall that $k^{1-\eps} = m^{\frac{1-\eps}{3}}$). 

We then use an additional randomization trick to ensure that any instance sampled from $\dist_1$ either has a ``large'' social welfare (w.p., say, half) or a ``small'' one (with the remaining probability):
we drop some of the bundles from the collection $\FF_\bid$ of each player $\bid \in N$ randomly (in a correlated way), to create two sub-distributions whereby in 
one of them none of  the special bundles are dropped and hence the social welfare is $k^3$, and in the other one all special bundles are dropped and hence the
social welfare is at most $k^{2+\eps}$ ($k^2$ for shared items plus $k^{\eps}$ intersection from any other bundle (in $\FF_\bid$) for each of the $k^2$ players). This completes the description of our hard distribution.  
We now formally define $\dist_1$. 

\textbox{Distribution $\dist_1(N,M)$. 
\textnormal{A hard input distribution for simultaneous protocols.}}
{

\smallskip

\textbf{Input:} Collections $N$ of $n=k^2$ players and $M$ of $m=2k^3$ items. \\
\textbf{Output:} A set of $n$ valuation functions $(v_1,\ldots,v_n)$ for the players in $N$.  
\\ 
\algline
\begin{enumerate}
	\item Let $\FS = \set{S_1,\ldots,S_p}$ be a $(p,q,t,\ell)$-intersecting family with $p = \exp\paren{\Theta(k^{\eps})}$, $q = k^2 + k$, $t=k$, and $\ell = k^{\eps}$ (guaranteed to exist
	 by Lemma~\ref{lem:intersecting-families}). 
	\item Pick $\jstar \in [p]$ and $\theta \in \set{0,1}$ independently and uniformly at random. 
	\item For each player $\bid \in N$ independently, 
	\begin{enumerate}
		\item Denote by $\FF_\bid$ the \emph{private collection} of player $\bid$ (used below to define the valuation function $v_\bid$), initialized to be 
		a copy of $\FS$ on the universe $[q]$. 
		\item Let $x_\bid \in \set{0,1}^{p}$ be a $p$-dimensional vector whereby $x_\bid(\jstar) = \theta$ and for any $j \neq \jstar$, $x_\bid(j)$ is chosen uniformly at random from $\set{0,1}$. 
		\item For any $j \in [p]$, if $x_\bid(j) = 0$, remove the set $S_j$ from $\FF_\bid$, and otherwise keep $S_j$ in $\FF_\bid$. 
	\end{enumerate}
	\item Pick a random permutation $\sigma$ of $M$. For the $\bid$-th player in $N$, map the $j$-th item in $[q] \setminus S_{\jstar}$ to $\sigma(j)$. Moreover, map the 
	$j$-th item in $S_{\jstar}$ to $\sigma(k^2+(\bid-1) \cdot k+j)$. Under this mapping, the private collection $\FF_\bid$ of player $\bid$ consists of at most $p$ sets of $t=k$ items from $M$.  
	\item For all $i \in N$, define the valuation function of player $\bid$ as $v_\bid(S) = \max_{T \in \FF_{\bid}} \card{S \cap T}$. 
\end{enumerate}
}

We use $\dist_1$ to denote the distribution $\dist_1(N,M)$ whenever the sets $N$ and $M$ are clear from the context (or are irrelevant). 
We make several observations about the distribution $\dist_1$. 

\begin{observation}\label{obs:XOS-valuation}
	The valuation function of each bidder $\bid \in [n]$ in the distribution $\dist_1$ is an XOS valuation (and hence is also subadditive) whereby each set $T \in \FF_{\bid}$ defines a clause in which all items in $T$ have 
	value $1$ and all other items have value $0$. 
\end{observation}

For any player $\bid \in N$, we define the \emph{labeling function} $\phi_\bid$ as the function used to map the items in $[q]$ to $M$. Notice that $\phi_\bid$ is a function of $\sigma$ and index $\jstar$.  

\begin{observation}\label{obs:dist1-input-def}
	The input to player $\bid$ can be uniquely identified by the pair $(x_\bid,\phi_\bid)$, as $x_\bid$ defines the private collection $\FF_\bid$ over the items $[q]$, and $\phi_\bid$ specifies the actual 
	labeling of the items in $M$ in the instance. 
\end{observation}

We also point out a crucial property of this distribution: each player $\bid \in N$ is \emph{oblivious} to which of the sets $S_j$ (for $j \in [p]$), is the set $S_{\jstar}$. More formally, 
\begin{observation}\label{obs:dist1-oblivious}
	Conditioned on the input $(x_\bid,\phi_\bid)$ to player $\bid$, the index $\jstar \in [p]$ is chosen uniformly at random. 
\end{observation}

Recall that for an instance $I \sim \dist_1$, $\sw(I)$ denotes the maximum value of social welfare, i.e., $\sw(I) := \max_{(A_1,\ldots,A_n)} \sum_{i \in N} v_i(A_i)$, where $(A_1,\ldots,A_n)$ ranges over all possible allocation of items. 
The following lemma establishes a bound on the social welfare of any instance sampled from $\dist_1$. 

\begin{lemma}\label{lem:sim-theta}
	For any $I \sim \distr{1}$, $(i)$ if $\theta = 1$, then $\sw(I) = k^3$, and $(ii)$ if $\theta = 0$, then $\sw(I) \leq 2k^{2+\eps}$. 
\end{lemma}
\begin{proof}
	Suppose first that $\theta = 1$. In this case, $x_\bid(\jstar) = 1$ for all bidders $\bid \in N$, implying that the set $S_{\jstar}$ is not removed from any private collection $\FF_\bid$. Moreover, the mapping $\sigma$ 
	maps the items in $S_{\jstar}$ to a unique set of items and hence the allocation $(A_1,\ldots,A_n)$, whereby $A_\bid$ is the set of items $\sigma(k^2+(\bid-1)\cdot k + 1) \ldots \sigma(k^2+\bid \cdot k)$, results 
	in a welfare of $k^3$ which is clearly maximum. 
	
	Now consider the case $\theta = 0$. In this case, $x_\bid(\jstar) = 0$ for all bidders $\bid \in N$, and hence the set $S_{\jstar}$ is missing from all private collections. Recall that items in $[q] \setminus S_{\jstar}$
	(across all players) are mapped to the first $k^2$ items of $M$ (according to the ordering $\sigma$). Moreover, by the intersecting family property of the set $\FS$, the intersection of $S_{\jstar}$ with any other
	 set in $\FS$, and consequently, any other set in any $\FF_\bid$ (for $\bid \in N$) is at most $\ell = k^{\eps}$ items. This means that in any allocation, bidder $\bid$ can only ``benefit'' from at most $k^{\eps}$ elements in
	 $\sigma(k^2+1) \ldots \sigma(k^3)$. Consequently, in this case, $\sw(I)$ is at most $k^2$ (accounting for all the first $k^2$ items of $\sigma$) plus $k^{2+\eps}$ (accounting for $k^{\eps}$ benefit from each of the 
	 $k^2$ players). 
\end{proof}

\subsection{The Lower Bound for Distribution $\distr{1}$}\label{sec:sim-lower}

Let $\prot$ be a public coin simultaneous protocol that can output a $\paren{m^{\frac{1-\eps}{3}}}$-approximation to the social welfare of any instance $I \sim \dist_1$, w.p. of failure $\delta \leq 1/3$. 
In this section, we prove that the communication cost of the protocol $\prot$ needs to be at least $\exp(k^{\Omega(\eps)})$ bits. Note that by (the easy direction of) Yao's minimax principle~\cite{Yao79}, 
we only need to consider deterministic protocols on the distribution $\dist_1$ to prove this result. 

The intuition behind the proof is as follows. By Lemma~\ref{lem:sim-theta}, the social welfare in the given instance changes by a factor of 
$k^{1-\eps}$ depending on the value of $\theta$. This implies that any $k^{1-\eps} = m^{\frac{1-\eps}{3}}$ approximation algorithm for the social welfare can also determine the value of $\theta$. 
Using this, we can argue that the message sent by the players needs to reveal $\Omega(1)$ bit of information about the parameter $\theta$. Roughly speaking, this means that each of the $n$ players
is responsible for revealing $\Omega(1/n)$ bit about $\theta$ in average.  

Furthermore, recall that the input to player $\bid \in N$ can be seen as a tuple $(x_\bid,\phi_\bid)$ (by Observation~\ref{obs:dist1-input-def}) and that $\theta = x_\bid(\jstar)$. Additionally, by
Observation~\ref{obs:dist1-oblivious}, given input $(x_\bid,\phi_\bid)$ to player $\bid$, the index $\jstar$ is chosen uniformly at random from $[p]$ and hence
player $\bid$ is oblivious to which index of $x_\bid$ corresponds to the parameter $\theta$. This essentially means that player $\bid$ needs to reveal $\Omega(p/n)$ bits about
the vector $x_\bid$ to be able to reveal $\Omega(1/n)$ bit about $x_\bid(\jstar)$, hence forcing $\bid$ to communicate $\Omega(p/n) = \exp\paren{k^{\Omega(\eps)}}$ bits also. To make the latter intuition precise, 
we argue that while the message sent by one player can, in principle, be used to infer information about the input of another player (as the input of the players are correlated), this extra information
is limited to an ``easy part'', containing only $(\sigma,\jstar)$ that can even be assumed to be known to referee (but not players) beforehand. This allows us to ``break'' the information
revealed to the referee to smaller pieces sent by each player, hence arguing that each player is indeed directly responsible for communicating the information about her input. 
We now formalize this intuition. We first need the following notation.

\paragraph{Notation.} We use $\Prot = (\Prot_1,\ldots,\Prot_n)$ to denote the random variable for the transcript of the messages communicated in $\prot$. 
For any player $\bid \in N$, and any $j \in [p]$, we use the random variable $X_{\bid,j} \in \set{0,1}$ to denote the value of $x_\bid(j)$, i.e., $X_{\bid,j} = 1$ iff the set $S_j \in \FS$ is included in the private collection $\FF_\bid$.  
We further define $X_\bid$ for  $\bid \in N$ as the vector $X_\bid := (X_{\bid,1},\ldots,X_{\bid,p})$. We use $\Sigma$ to denote the random variable
 for the permutation $\sigma$, $J$ for the index $\jstar$, and $\Theta$ for the parameter $\theta$. For each player $\bid \in N$, $\Phi_\bid$ denotes the random variable for  the labeling function $\phi_\bid$.

Recall that $(\Sigma,J)$ is the ``easy part'' of the input: the part that we assume the referee (but not each individual player) knows beforehand. Assuming this knowledge can only strengthen our lower bound. 
We start by arguing that the protocol $\prot$ needs to reveal $\Omega(1)$ bits of information about the value of parameter $\theta$ in the distribution.

\begin{claim}\label{clm:sim-theta-discovery}
	$\mi{\Theta}{ \Prot  \mid \Sigma,J} = \Omega(1)$. 
\end{claim}
\begin{proof}
	By Claim~\ref{lem:sim-theta}, the social welfare is $k^{1-\eps} = m^{\frac{1-\eps}{3}}$ times larger when $\theta = 1$ than when $\theta = 0$. Since $\prot$ outputs an $\paren{m^{\frac{1-\eps}{3}}}$-approximation to 
	the social welfare, it can also be used to distinguish between the values of $\theta$ w.p. of error at most $\delta \leq 1/3$. This means that there is a function that given the
	message $\Prot$, and variables $(\Sigma,J)$ (i.e., the easy part of the input) can
	 determine the value of $\Theta$ w.p. of error at most $\delta$. This, together with Fano's inequality (Fact~\ref{fact:fano}), implies that $\HH(\Theta \mid \Prot,\Sigma,J) \leq H_2(\delta)$ (as $\card{\Theta} = 2$). 
	
	We now have,
	\begin{align*}
		H_2(\delta) &\geq \HH(\Theta \mid \Prot,\Sigma,J) = \HH(\Theta \mid \Sigma,J) - \mi{\Theta}{\Prot \mid \Sigma,J} = 1- \mi{\Theta}{\Prot \mid \Sigma,J}
	\end{align*}
	where in the final equality we used the fact that in $\dist_1$, $\Theta$ is chosen uniformly at random from $\set{0,1}$ independent of $(\Sigma,J)$, and hence $\HH(\Theta \mid \Sigma,J) = 1$ (by \itfacts{uniform}). To finalize, 
	we have that $\mi{\Theta}{\Prot \mid \Sigma,J} \geq 1-H_2(\delta) = \Omega(1)$ as $\delta$ is a constant bounded away from $1/2$. 
\end{proof}

We now show that the information revealed about $\Theta$ by the message $\Prot$ is at most the sum of information revealed by each message $\Prot_\bid$ for $\bid \in N$ individually. In other words, one does not gain an
extra information by combining the messages of players (after conditioning on what is revealed by $(\Sigma,J)$ already). 

\begin{claim}\label{clm:sim-independent-message}
	$\mi{\Theta}{\Prot \mid \Sigma,J} \leq \sum_{\bid \in N} \mi{\Theta}{\Prot_\bid \mid \Sigma,J}$. 
\end{claim}
\begin{proof}
	We have,
	\begin{align*}
		\mi{\Theta}{\Prot \mid \Sigma,J} &= \sum_{\bid \in N} \mi{\Theta}{\Prot_\bid \mid \Prot^{<\bid},\Sigma,J} \leq \sum_{\bid \in N} \mi{\Theta}{\Prot_\bid \mid \Sigma,J} 
	\end{align*}
	where the equality is by chain rule (\itfacts{chain-rule}), and the inequality follows from Proposition~\ref{prop:info-decrease}, as we show below that
	$\Prot_\bid \perp \Prot^{<\bid} \mid \Theta,\Sigma,J$, or equivalently $\mi{\Prot_\bid}{\Prot^{<\bid} \mid \Theta,\Sigma,J} = 0$ (by~\itfacts{info-zero}). 
	
	As stated in Observation~\ref{obs:dist1-input-def}, the input of player $\bid \in N$ is uniquely determined by $(x_\bid,\phi_\bid)$ and hence $\Prot_\bid$ is a deterministic 
	function of variables $X_\bid$ and $\Phi_\bid$. Moreover, $\Phi_\bid$ is also uniquely determined by $(\Sigma,J)$, hence, conditioned on $(\Sigma,J)$, $\Prot_\bid$ is only a function of $X_\bid$. 
	On the other hand, conditioned on $(\Theta,\Sigma,J)$, $X_\bid$ and
	$X^{<\bid}$ are chosen independently of each other in the distribution $\dist_1$ (as $X_{\bid,\jstar} = \theta$ and the rest of $X_\bid$ is chosen uniformly at random from $\set{0,1}$). 
	This implies that $\mi{X_\bid}{X^{<\bid} \mid \Theta,\Sigma,J} = 0$.  As stated earlier, $\Prot_\bid$ is a function of 
	$X_\bid$ and $\Prot^{<\bid}$ is a function of $X^{<\bid}$ alone (conditioned on $(\Theta,\Sigma,J)$), hence, by data processing
	inequality (\itfacts{data-processing}), $\mi{\Prot_\bid}{\Prot^{<\bid} \mid \Theta,\Sigma,J} = 0$ as well. 
\end{proof}

We now use a direct-sum style argument to prove that if a player $\bid \in N$ wants to communicate $c$ bits about $\theta$, she needs to communicate (essentially) $p \cdot c$ bits
about her input.

\begin{lemma}\label{lem:sim-direct-sum}
	For any $\bid \in N$, $\mi{\Theta}{\Prot_\bid \mid \Sigma,J} \leq \card{\Prot_\bid}/p$. 
\end{lemma}
\begin{proof}
	We have, 
	\begin{align*}
		\mi{\Theta}{\Prot_\bid \mid \Sigma,J} &= \Ex_{j \in [p]}\Bracket{\mi{\Theta}{\Prot_\bid \mid \Sigma,J=j}} = \Ex_{j \in [p]} \Bracket{\mi{X_{\bid,j}}{\Prot_\bid \mid \Sigma,J=j}} \tag{$\Theta = X_{\bid,j}$ conditioned on $J=j$} \\
		&= \frac{1}{p} \cdot \sum_{j \in [p]} \mi{X_{\bid,j}}{\Prot_\bid \mid\Sigma,J=j} \tag{the index $\jstar$ is chosen uniformly at random from $[p]$}
	\end{align*}
	
	Define $\Sigma^{-\bid}$ as the part of permutation $\Sigma$ that does not affect the labeling function $\Phi_\bid$ of player $\bid$, i.e., the values of $\sigma(k^2+1)\ldots \sigma(k^2+(\bid-1)\cdot k)$ and 
	$\sigma(k^2+\bid\cdot k+1),\ldots,\sigma(k^3)$. With this notation, $\Sigma$ can be written as a function of $\Phi_\bid$, $\Sigma^{-\bid}$, and $J$ (as $J$ and $\Phi_\bid$ uniquely define
	the rest of $\Sigma$ outside $\Sigma^{-\bid}$). 
	Consequently, we can write, 
	\begin{align*}
		\mi{\Theta}{\Prot_\bid \mid \Sigma,J} &= \frac{1}{p} \cdot \sum_{j \in [p]} \mi{X_{\bid,j}}{\Prot_\bid \mid \Sigma^{-\bid},\Phi_\bid,J=j}
	\end{align*}
	
	Our goal is now to drop the conditioning on the event ``$J = j$''. To do so, notice that the distribution of $(\Sigma^{-\bid},\Phi_\bid)$ is independent of the event $J=j$; this is immediate to see as $\Sigma^{-\bid}$
	is independent of $\Phi_\bid$ and $J=j$, and $\Phi_\bid$ is independent of $J=j$ by Observation~\ref{obs:dist1-oblivious}. Moreover, $X_{\bid,j}$ is independent of all $(\Sigma^{-\bid},\Phi_\bid,J=j)$ (as it is uniform
	over $\set{0,1}$) and furthermore, $\Prot_\bid$ is a function of $\Phi_\bid,X_\bid$, which are independent of $J=j$. Consequently, we can drop the conditioning in the above information term and obtain that, 
	\begin{align*}
		\mi{\Theta}{\Prot_\bid \mid \Sigma,J} &= \frac{1}{p} \cdot \sum_{j \in [p]} \mi{X_{\bid,j}}{\Prot_\bid \mid \Sigma^{-\bid},\Phi_\bid} 
		\leq \frac{1}{p} \cdot \sum_{j \in [p]} \mi{X_{\bid,j}}{\Prot_\bid \mid X_\bid^{<j},\Sigma^{-\bid},\Phi_\bid} \tag{by Proposition~\ref{prop:info-increase} as $X_{\bid,j} \perp X_{\bid}^{<j} \mid \Sigma^{-\bid},\Phi_\bid$} \\
		&= \frac{1}{p} \cdot \mi{X_\bid}{\Prot_\bid \mid \Sigma^{-\bid},\Phi_\bid} \leq \frac{1}{p} \cdot \HH(\Prot_\bid \mid \Sigma^{-\bid},\Phi_\bid) \leq \frac{1}{p} \cdot \HH(\Prot_\bid) \leq \frac{1}{p} \cdot \card{\Prot_\bid} 
	\end{align*}
	where the equality in the second line is by chain rule (\itfacts{chain-rule}), and inequalities are by \itfacts{uniform} and \itfacts{cond-reduce}.
\end{proof}

We can now conclude the following lemma. 
\begin{lemma}\label{lem:sim-conclude}
	Communication cost of $\prot$ is $\Omega(p)$.  
\end{lemma}
\begin{proof}
		$\norm{\prot} = \sum_{\bid \in N} \card{\Prot_\bid} \geq p \cdot  \sum_{\bid \in N} \mi{\Theta}{\Prot_\bid \mid \Sigma,J} \geq p \cdot \mi{\Theta}{\Prot \mid \Sigma,J} = \Omega(p)$. 
	where the last three equations are by, respectively, Lemma~\ref{lem:sim-direct-sum}, Claim~\ref{clm:sim-independent-message}, and Claim~\ref{clm:sim-theta-discovery}. 
\end{proof}

Theorem~\ref{thm:sim-lower} now follows from Lemma~\ref{lem:sim-conclude} by re-parameterizing $\eps$ above by some $\Theta(\eps)$ and noting
that $p = \exp\paren{\Theta(k^{\eps})} = \exp\paren{m^{\Omega(\eps)}}$ (as $m = k^3$).

%% file: multi-round.tex
\newcommand{\Rpri}{\ensuremath{R_{\textnormal{\textsf{pri}}}}}
\newcommand{\Rpub}{\ensuremath{R_{\textnormal{\textsf{pub}}}}}

\newcommand{\FC}{\ensuremath{\mathcal{C}}}

\section{Main Result: A Lower Bound for Multi-Round Protocols}\label{SEC:MULTI}
In this section, we establish our main result. Formally, 

\begin{theorem}\label{thm:multi-lower}
	For any integer $1 \leq r \leq o\paren{\frac{\log{m}}{\log\log{m}}}$, and any sufficiently small constant $\eps > 0$, any $r$-round protocol (possibly randomized) for combinatorial auctions with subadditive (even XOS) bidders 
	that can approximate the value of social welfare to a factor of $\paren{\frac{1}{r} \cdot m^{\frac{1-\eps}{2r+1}}}$ requires $\exp\paren{m^{\Omega(\frac{\eps}{r})}}$ bits of communication. 
\end{theorem}

We start by introducing the recursive family of hard input distributions that we use proving in Theorem~\ref{thm:multi-lower} and then establish a lower 
bound for this distribution. 

\subsection{A Hard Input Distribution for $r$-Round Protocols} \label{sec:multi-dist}

Our hard distribution $\dist_r$ for $r$-round protocols is defined recursively with its base case ($r=1$ case) being the distribution $\dist_1$ introduced in Section~\ref{sec:sim-dist}.
We first give an informal description of $\dist_r$. 

Let $k$ be an integer and consider a set $N$ of $n_r=k^{2r}$ players and a set $M$ of $m_r=(r+1)\cdot k^{2r+1}$ items. The players are partitioned (arbitrary) between $k^2$ groups 
$N_1,\ldots,N_{k^2}$ each of size $n_{r-1}$. Fix a group $N_g$ and for any player $\bid \in N_g$, we create an exponentially large (in $k$) collection $\FC_i$ of item-sets of size $m_{r-1}$ (over the universe $M$), such that the 
for any two sets $S,T \in \FC_i$, $\card{S \cap T} \leq k^{2r-2+\eps}$ (for any constant $\eps > 0$). 

The \emph{local} view of player $\bid \in N_g$ is as follows: over each set $S_j \in \FC_\bid$, we create an $(r-1)$-round instance of the problem, namely instance $I_{\bid,j}$, sampled from the distribution $\dist_{r-1}$ with the 
set of players being $N_g$ and the set of items being $S_j$, and then let the input of player $\bid$ be the collective input of the $\bid$-th player in all these instances. In other words, player $\bid$ finds herself 
``playing'' in exponentially many ``$(r-1)$-round instances'' of $\dist_{r-1}$.

On the \emph{group level}, the input to players inside a group $N_g$ are highly correlated: for each player $\bid \in N_g$, one of the instances, namely $I_{\bid,\jstar}$, is an ``special instance'' in the sense that 
\emph{all} players in the group $N_g$ has a ``consistent'' view of this instance, i.e., the collective view of players $1,\ldots,n_{r-1}$ in $N_g$ on the instances $I_{1,\jstar},\ldots,I_{n_{r-1},\jstar}$ forms
a valid instance sampled from $\dist_{r-1}$. However, for any other index $j \neq \jstar$, the collective view of players in $N_g$ in the instances $I_{1,\jstar},\ldots,I_{n_{r-1},\jstar}$ 
forms a ``pseudo instance'' that is \emph{not} sampled from $\dist_{r-1}$; these pseudo instances are created by sampling the input of each player \emph{independently} according to $\dist_{r-1}$. Note however that
while the pseudo instances and the special instance of a player are fundamentally different, each player is oblivious to this difference, i.e., which instance is the special instance. 

Finally, the input to players across the groups, i.e., the \emph{global} input, is further correlated: the set of items in the special instances of players in a group $N_g$ is a ``unique'' set of items (across all groups), while
\emph{all} other instances, across all groups, are constructed over a set of $k^{2r}$ ``shared'' items. This correlation makes the special instance of a player $i$, in some sense, the \emph{only important} instance: to obtain a large allocation, the players need to ultimately solve the problem for these special instances.

We now formally define distribution $\dist_{r}$. In the following, for simplicity of exposition, we assume that the distribution $\dist_r$, in addition to the valuation function of players, also
outputs the \emph{private collections} (defined similarly as in $\dist_1$) of players
that are used to define these functions\footnote{Strictly speaking, this is a redundant information as the valuation functions can uniquely determine the private collections; however, we include this redundant
output for the ease of presentation.}. 

\textbox{Distribution $\dist_r(N,M)$. 
\textnormal{A hard input distribution for $r$-round protocols (for $r \geq 2$).}}
{

\smallskip

\textbf{Input:} Collections $N$ of $n_r=k^{2r}$ players and $M$ of $m_r= (r+1) \cdot k^{2r+1}$ items. \\
\textbf{Output:} A set of $n_r$ valuation functions $(v_1,\ldots,v_{n_r})$ for the players in $N$ and $n_r$ private collections $(\FF_1,\ldots,\FF_{n_r})$ used to define the valuation functions.   
\\ 
\algline
\begin{enumerate}
	\item Let $\FS_r = \set{S_{1},\ldots,S_{p}}$ be a $(p_r,q_r,t_r,\ell_r)$-intersecting family with
	 parameters $p_r = p = \exp\paren{\Theta(k^{\eps})}$, $q_r = k^{2r}+r \cdot k^{2r-1}$, $t_r = r \cdot k^{2r-1}$, and $\ell_r = k^{2r-2 + \eps}$ (guaranteed to exist by Lemma~\ref{lem:intersecting-families} 
	 as $k = m^{\Omega(1/r)} = \omega(r^{2/\eps})$ by the assumption that $r = o\paren{\frac{\log{m}}{\log\log{m}}}$). 
 
	\item Arbitrary group the players into $k^2$ groups $\FN = (N_1,\ldots,N_{k^2})$, whereby each group contains exactly $n_{r-1} = k^{2r-2}$ players. 	
	\item Pick an index $\jstar \in [p]$ uniformly at random and sample an instance $\Istar_r \sim \dist_{r-1}([n_{r-1}],S_{\jstar})$. 
	\item For each group $N_g \in \FN$ independently,
	\begin{enumerate}
		\item \label{line:copy} Define $\Istar_{N_g}$ as $\Istar_r$ by mapping the players in $[n_{r-1}]$ to $N_g$. 
		\item \label{line:instances} For each player $\bid \in N_g$ \emph{independently}, create $p$ instances $I^{(\bid)}:= (I_{\bid,1},\ldots,I_{\bid,p})$ whereby
		 for all $j \neq \jstar$, $I_{\bid,j} \sim \dist_{r-1}(N_g,S_j)$, and $I_{\bid,\jstar} = \Istar_{N_g}$. 
		\item For a player $\bid \in N_g$ and index $j \in [p]$, let $\FF_{\bid,j}$ be the set of \emph{private collection} of that player in instance $I_{\bid,j}$ and let $\FF_\bid = \bigcup_{j \in [p]} \FF_{\bid,j}$. 
	\end{enumerate} 
	\item Pick a random permutation $\sigma$ of $M$. For each $g \in [k^2]$ and group $N_g$, map the $k^{2r}$ items in $[q_r] \setminus S_{\jstar}$ to $\sigma(1),\ldots,\sigma(k^{2r})$, and
	the $t_r$ items in $S_{\jstar}$ to $\sigma((g-1) \cdot t_r+1) \ldots \sigma(g \cdot t_r)$ (and for each player $\bid \in N_g$, update the item set of $\FF_\bid$ and
	underlying instances $I_{\bid,1},\ldots,I_{\bid,p}$ accordingly). 
	\item For any player $\bid \in N$, define the valuation function of player $\bid$ as $v_\bid(S) = \max_{T \in \FF_\bid} \card{S \cap T}$ (note that these valuation functions
	 are XOS valuation; see Observation~\ref{obs:XOS-valuation}). 
\end{enumerate}
}

We make several observations about the distribution $\dist_r$. Recall that $\FF_\bid$ denotes the private collection of player $\bid \in N$ that is used to define the valuation function $v_\bid$. By construction, the size of the sets inside 
each private collection is equal across any two distributions $\dist_{r}$ and $\dist_{r'}$ and hence is equal to $k$ (by definition of distribution $\dist_1$). A simple property of these sets is that, 
\begin{observation}\label{obs:distr-uniform-set}
	For any player $\bid \in N$, and any set $T \in \FF_\bid$, the set $T$ is chosen uniformly at random from all $k$-subsets of $M$. 
\end{observation}

Fix any group $N_g \in \FN$ and any player $\bid \in N_g$. The input to player $\bid$ can be seen as the ``view'' of $\bid$ in the $p$ instances $I^{(\bid)} := (I_{\bid,1},\ldots,I_{\bid,p})$, i.e., 
the input of the $\bid$-th player (in $N_g$) in $I_{\bid,j}$ (for all $j \in [p]$) and \emph{not} the whole instance. However, in the following, we slightly abuse the notation and use $I_{\bid,j}$ to also denote the
view of player  $\bid$ in the instance $I_{\bid,j}$. Moreover, we point out that $I_{\bid,j}$ is defined over the set of items $S_j$; hence, the complete input to player $\bid$ 
is the pair $(I^{(\bid)},\phi_\bid)$ where $\phi_\bid$ is the labeling function to map the items in $S_j$ to $M$ (see also Observation~\ref{obs:dist1-input-def}).

For any player $\bid \in N$, we refer to the instance $I_{\bid,\jstar}$ of player $\bid$ as the \emph{special instance} of player $\bid$, and to all other instances $I_{\bid,j}$ for $j \neq \jstar$ as 
\emph{fooling instances}. 

\begin{observation}\label{obs:distr-independence}
	For any group $N_g \in \FN$, the joint input of all players $\bid \in N_g$ in their special instances $I_{\bid,\jstar}$ form the instance $\Istar_{N_g}$ that is sampled from the distribution $\dist_{r-1}$. 
\end{observation}

On the other hand, the fooling instances of players $\bid \in N_g$ are sampled \emph{independently} and hence the joint distribution of the players on their instances $I_{\bid,j}$ is \emph{not} sampled from
$\dist_{r-1}$.  Nevertheless, this difference is \emph{not evident} to the player $\bid$. 

\begin{observation}\label{obs:distr-j-independence}
	For any player $\bid \in N$, conditioned on the input $(I^{(\bid)},\phi_\bid)$ given to the player $\bid$, the index $\jstar$ is chosen 
	uniformly at random from $[p]$. 
\end{observation}

Additionally, 
\begin{observation}\label{obs:distr-input-independence}
	The distribution of collection of instances $\Ins:=(I^{(1)},\ldots,I^{(n_r)}) \sim \dist_r \mid \Istar_r,\sigma,\jstar$ is a \emph{product} distribution as instances in Line~(\ref{line:instances}) are sampled independently (except for
	 instances $I_{\bid,\jstar} = \Istar_r$ which are already conditioned on above). 
\end{observation}

Another important property of the special instances in distribution $\dist_r$ is that,
\begin{observation}\label{obs:distr-disjoint}
	The special instances $\Istar_{N_1},\ldots,\Istar_{N_{k^2}}$ are supported on disjoint set of items (according to the mapping $\sigma$). 
\end{observation}

Notice that we can trace the special instances into a \emph{unique} path $\Istar_{r} \rightarrow \Istar_{r-1} \rightarrow \ldots \rightarrow \Istar_{2}$, whereby $\Istar_{2}$ is sampled from the distribution $\dist_1$. 
We use $\thetas$ to denote the parameter $\theta$ (in $\dist_1$) in the instance $\Istar_{2}$ in this path. The following lemma proves a key relation between $\thetas$ and social welfare of the sampled instance. 

\begin{lemma}\label{lem:multi-theta}
	For any instance $I \sim \dist_r$: 
	\begin{align}
		&\Pr\paren{\sw(I) \geq k^{2r+1} \mid \thetas = 1} = 1 \label{eq:thetas-1} \\
		&\Pr\paren{\sw(I) \leq 2r \cdot k^{2r+2\eps} \mid \thetas = 0} = 1- r \cdot \exp\paren{-\Omega(k^\eps)} \label{eq:thetas-0} 
	\end{align}
\end{lemma}
\begin{proof}
	We start by the simpler case of Eq~(\ref{eq:thetas-1}); the proof is by induction. The base case, when $r=1$, is true by Lemma~\ref{lem:sim-theta}. 
	Suppose this holds for all integers smaller than $r$. Now, consider an instance $I \sim \paren{\dist_{r} \mid \thetas= 1}$ and the $k^2$ special
	instances $I_{N_1}, \ldots, I_{N_{k^2}}$ sampled from $\paren{\dist_{r-1} \mid \thetas=1}$ in $I$. By induction, there is an allocation $A_g$ for each $g \in [k^2]$ 
	that results in a welfare of at least $k^{2r-1}$ in each $\Istar_{N_g}$. By Observation~\ref{obs:distr-disjoint}, the set of items
	 among special instances are disjoint, and hence the allocation $A:= (A_1,\ldots,A_{k^2})$ which assigns the 
	bundles in $A_g$ to players in $N_g$ for $g \in [k^2]$ is a valid allocation that results in a welfare of $k^2 \cdot k^{2r-1} = k^{2r+1}$, proving the induction step. 
	
	We now prove Eq~(\ref{eq:thetas-0}) by induction. The base case of $r=1$ is true by Lemma~\ref{lem:sim-theta}. Assume that the bounds hold for all integers smaller than $r$ and consider an
	instance $I \sim \paren{\dist_{r} \mid \thetas= 0}$ and let $\Istar_{N_1},\ldots,\Istar_{N_{k^2}}$ be the special instances of $I$, ``copied'' from
	the instance $\Istar_r \sim \paren{\dist_{r-1} \mid \thetas =0}$ (as in Line~(\ref{line:copy}) of $\dist_r$). Let $U$ be the set of items assigned to these instances (by mapping $\sigma$) and 
	$\neg{U}$ be the set of remaining items assigned by $\sigma$, i.e., the items that have no value in the special instances; we have
	 $\card{U} = k^{2} \cdot t_r = r \cdot k^{2r+1}$ and $\card{\neg{U}} = k^{2r}$ (notice that $\sigma$ does \emph{not} assign all the items in $M$; in particular,
	$k^{2r+1}-k^{2r}$ items are not assigned to any instance, i.e., have no value for any player; these extra items are only added to simplify the math.). 
	We have, 
	
	\begin{claim}\label{clm:k-eps-intersect}
		W.p. $1-\exp\paren{-\Omega(k^{\eps})}$, for any player $\bid \in N$ and any set $T \in F_\bid$ such that $T$ does \emph{not} belong to a private collection of a special instance (i.e., 
		$T$ is not sampled from $I_{\bid,\jstar}$), $\card{T \cap U} \leq k^{2\eps}$.  
	\end{claim}
		\begin{proof}
		Fix a group $N_g \in \FN$ and fix a player $\bid \in N_g$ and let $I_{\bid,j}$ be an instance of $\dist_{r-1}$ for some $j \neq \jstar$, i.e., not a special instance. Recall that the set of items in $I_{\bid,j}$ and 
		$I_{\bid,\jstar}$ are two distinct sets $S_j$ and $S_{\jstar}$ from $\FS_r$ on the universe $[q_r]$ (and hence $\card{S_j \cap S_{\jstar}} \leq \ell_r = k^{2r-2+\eps}$ by definition of intersecting families), and
		since $[q_r]$ is entirely mapped by $\sigma$ for player $\bid \in N_g$, the intersection between item set of $I_{\bid,j}$ and $I_{\bid,\jstar}$ is at most $k^{2r-2+\eps}$; this in particular means
		that at most $k^{2r-2+\eps}$ items in $I_{\bid,j}$ belong to $U$ ($I_{\bid,j}$ does not share any item with any instance $I_{\bid',\jstar}$ for any $\bid' \notin N_g$). 
		
		Now consider the choice of a set $T$ (in the private collection) for the player $\bid$ in the
		instance $I_{\bid,j}$. For each item $a$ that belongs to both item-set of $I_{\bid,j}$ and $U$, define an indicator random variable $X_{a} \in \set{0,1}$, 
		which is one iff $a$ is chosen in $T$. Then, $X:= \sum_{a} X_a$ denotes $\card{T \cap U}$. By Observation~\ref{obs:distr-uniform-set}, $T$ is a $k$-subset chosen uniformly
		at random from a universe of size $t_r = r \cdot k^{2r-1}$, and hence, $\Ex\bracket{X} \leq k^{2r-2+\eps} \cdot 1/(r \cdot k^{2r-2}) \leq k^{\eps}/r$. 
		By Chernoff bound for negatively correlated random variables (Proposition~\ref{prop:chernoff}), 
		$\Pr\paren{\card{S \cap U} \geq k^{2\eps}} \leq \exp\paren{{-\Omega(k^{2\eps})}}$. 
		
		We can now apply a union bound for all possible choices for the set $T$ (among all players and instances), and the probability that even one set $T$ violates this constraint is (note that there
		 are $n_r \cdot p^r$ different choices for $T$)
		\begin{align*}
			n_r \cdot p^r \cdot \exp\paren{{-\Omega(k^{2\eps})}} = \exp\paren{\Theta(r \cdot \log{k})} \cdot \exp\paren{\Theta(r \cdot k^{\eps})}\cdot \exp\paren{{-\Omega(k^{2\eps})}} = \exp\paren{-\Omega(k^{\eps})} 
		\end{align*}
		since $r = o(k^{\eps})$ (by the assumption that $r = o\paren{\frac{\log{m}}{\log\log{m}}}$). 
	\end{proof}

	In the following we condition on the event in Claim~\ref{clm:k-eps-intersect} (event $\event_1$) and the event that $\sw(\Istar) \leq 2(r-1) \cdot k^{2r-2+2\eps}$ (event $\event_2$).
	Note that by Claim~\ref{clm:k-eps-intersect} and induction hypothesis, these two events happen (simultaneously) w.p. $1 - r \cdot \exp\paren{-\Omega(k^\eps)}$. 
	
	Now fix any allocation $\FA  = (A_1,\ldots,A_n)$. As size of $\neg{U}$ is at most $k^{2r}$, the items in $\neg{U}$ can only contribute $k^{2r}$ to the welfare in $\FA$. Next, let $\FA^*$ be the 
	subset of $\FA$ such that the maximizing clause in each $A_i \in \FA^*$ (i.e., the set $T \in \FF_i$) belongs to some special instance, and $\FA'$ be the remaining part of allocation $\FA$. 	
	We know, by $\event_2$, that the contribution of $\FA^*$ to the welfare is at most $k^2 \cdot 2(r-1) \cdot k^{2r-2+2\eps} = 2(r-1) \cdot k^{2r+2\eps}$ (counting the $k^2$ special instances). Moreover, 
	by $\event_1$ (in Claim~\ref{clm:k-eps-intersect}), the contribution of $\FA'$ is at most $k^{2r} \cdot k^{2\eps} = k^{2r+2\eps}$. To conclude, we obtain that the
	social welfare when $\thetas=0$ is at most $k^{2r} + 2(r-1) \cdot k^{2r+2\eps} + k^{2r+2\eps} \leq 2r \cdot k^{2r+2\eps}$ with the desired probability, proving the lemma. 
\end{proof}

\subsection{The Lower Bound for Distribution $\dist_r$}\label{sec:multi-lower}

Let $\prot$ be a $r$-round protocol that can output a $\paren{\frac{1}{r} \cdot m_r^{\frac{1-2\eps}{2r+1}}}$-approximation to the social welfare of any instance $I \sim \dist_r$, w.p. of failure $\delta < 1/4$.
In this section, we prove that the communication cost of the protocol $\prot$ needs to be at least $\exp(\Omega(k^{\eps}))$ bits. By (the easy direction of) Yao's minimax principle~\cite{Yao79}, it suffices
to prove this lower bound for deterministic algorithms.   

We start by providing a detailed overview of the proof.  First, by Lemma~\ref{lem:multi-theta} we can argue that the protocol $\prot$ is also a $\paren{\delta+o(1)}$-error protocol for estimating 
the parameter $\thetas$, and hence we prove the lower bound for $\thetas$-estimation problem instead.  Recall that in any instance  $I_r \sim \dist_r$, the value of $\thetas$ is equal to the value of $\thetas$ in
the underlying special instance $\Istar_r$ in $I_r$, and that $\Istar_r$ is sampled from the distribution $\dist_{r-1}$. Hence to ``solve'' the instance $I_r \sim \dist_{r}$, the players need to be able 
to solve the instance $\Istar_r \sim \dist_{r-1}$ as well. This suggests an inductive approach to prove the lower bound for the distribution $\dist_r$. 

Consider the first message $\Prot_1 = (\Prot_{1,1},\ldots,\Prot_{1,n_r})$ of $\prot$. Recall that the input to any player $\bid \in N$ consists of $p$ different instances (of $\dist_{r-1}$), 
one of which being the instance $\Istar_r$. By Observation~\ref{obs:distr-j-independence}, each player $\bid$ is oblivious to the identity of 
$\Istar_r$ and hence, intuitively, the message $\Prot_{1,i}$ cannot reveal more than $\approx \card{\Prot_{1,i}}/p$ bits of information about the instance $\Istar_r$. Considering the simultaneity of the protocol $\prot$, 
we can use a similar argument as in the previous section and prove that if $\card{\Prot_1} = o(p)$, then at most $o(1)$ bits of information is revealed about $\Istar_r$. 

Now consider the second round of the protocol $\prot$. The task of players in each group $N_g \in \FN$ is now to solve the instance $\Istar_r$ (on a separate set of players and items). 
As argued above, the first message of players can only reveal $o(1)$ bits of information about $\Istar_r$ and hence distribution of $\Istar_r$ is still ``very close'' to its original distribution $\dist_{r-1}$, 
\emph{even conditioned on the first message of players}. But $\dist_{r-1}$ is assumed inductively to be a hard input distribution for $(r-1)$-round protocols
and as $\prot$ needs to solve $\Istar_r$ in $(r-1)$ rounds now, we may argue that it needs an exponential communication. 

To make this intuition precise, we employ a \emph{round-elimination} argument: Given any hard instance $I_{r-1} \sim \dist_{r-1}$, we ``embed'' $I_{r-1}$ 
in an $r$-round instance $I_{r}$ sampled from $\dist_r$ \emph{conditioned on the first message $\Prot_1$} of $\prot$ with \emph{no communication} between the players and then use $\prot$ from the second round onwards
to solve $I_{r-1}$. However, notice that as the number of players (and items) vary between $I_r$ and $I_{r-1}$, we cannot directly apply $\prot$ on $I_{r-1}$. 
Instead, the players first sample a message $\Prot_1$ (of $\prot$) according to the distribution $\dist_r$ using public randomness. Next, each player $i \in [n_{r-1}]$ in the 
instance $I_{r-1}$ mimics the role of $k^2$ different players (one ``copy'' in each group in $\FN$ in $I_{r}$) by letting the input of each copy in the special instance (of $I_r$) be her input in $I_{r-1}$ and
 then ``completes'' the rest of her input (i.e., her fooling instances in $I_r$) \emph{independently} of other players to obtain an instance $I_r \sim \dist_r \mid \Istar_r=I_{r-1},\Prot_1$. 
Note that a-priori it is not clear that why such an embedding is possible since the first message $\Prot_1$ correlates the input of players in fooling instances, making independent sampling of these instances impossible.
However, we show that by further conditioning on some ``easy part'' of the input in the first round, i.e., $\sigma$ and $\jstar$ (by sampling these parts publicly also), the players can indeed implement this embedding \emph{without any communication} and hence obtain a valid $(r-1)$-round protocol for $I_{r-1}$. We are now ready to present the formal proof. To continue, we need the following notation.

\paragraph{Notation.} For any $j \in [r]$, we use $\Prot_j = (\Prot_{j,1},\ldots,\Prot_{j,n_r})$ to denote the random variable for the transcript of the messages communicated in the round $j$ of $\prot$.
For any player $\bid \in N$, and any $j \in [p]$, we override the notation and use $I_{\bid,j}$ to also denote the random variable for the instance $I_{\bid,j}$ sampled in $\dist_r$ (similarly
 for $\Istar_r$ and $I^{(\bid)}$). We further use $\Sigma$ to denote the random variable
for the permutation $\sigma$ and $J$ for the index $\jstar$. 
We start by the following simple claim. 
\begin{claim}\label{clm:multi-theta-apx}
	Protocol $\prot$ can also determine the value of $\thetas$ w.p. $1-\delta-o(1)$. 
\end{claim}
\begin{proof}
	By Lemma~\ref{lem:multi-theta}, the ratio of $\sw(I)$ depending on the parameter $\thetas$ is (w.p. $1-o(1)$): 
	\begin{align*}
		\frac{k^{2r+1}}{2r \cdot k^{2r+2\eps}} = \frac{k^{1-2\eps}}{2r} = \frac{m_r^{\frac{1-2\eps}{2r+1}}}{2r \cdot (r+1)^{\frac{1}{2r+1}}} > \frac{1}{r} \cdot m_r^{\frac{1-2\eps}{2r+1}}
	\end{align*}
	Hence, the $\delta$-error $\paren{\frac{1}{r} \cdot m_r^{\frac{1-2\eps}{2r+1}}}$-approximation protocol $\prot$ correctly determines the value of $\thetas$ w.p. $1-\delta-o(1)$. 
\end{proof}

We show that as long as the first message sent by the players is not too large, 
this message cannot reveal much information about the special instance $\Istar_r$ embedded in the distribution $\dist_r$.  
This argument is a similar to the one in Section~\ref{sec:sim-lower}.

\begin{lemma}\label{lem:small-info}
	If $\card{\Prot_1} = o(p/r^4)$, then $\mi{\Istar_r}{\Prot_1 \mid \Sigma,J} = o(1/r^4)$.
\end{lemma}

We break the proof of Lemma~\ref{lem:small-info} into two separate steps. First, we argue that the information revealed by the first message about $\Istar_r$ 
is at most the total summation of the information revealed by each individual player about $\Istar_r$, \emph{after} conditioning on the ``easy part'' of the input in the first round, 
i.e., $\sigma$ and $\jstar$.

The proof of this claim is essentially identical to that of Claim~\ref{clm:sim-independent-message} and is provided for completeness. 

\begin{claim}\label{clm:multi-independent-message}
	$\mi{\Istar_r}{\Prot_1 \mid \Sigma,J} \leq \sum_{\bid \in N} \mi{\Istar_r}{\Prot_{1,\bid} \mid \Sigma,J}$. 
\end{claim}
\begin{proof}
	We have, 
	\begin{align*}
		\mi{\Istar_r}{\Prot_1 \mid \Sigma,J} &= \sum_{\bid \in N} \mi{\Istar_r}{\Prot_{1,\bid} \mid \Prot_1^{<\bid},\Sigma,J} \leq \sum_{i \in N} \mi{\Istar_r}{\Prot_{1,\bid} \mid R,\Sigma,J}
	\end{align*}
	where the equality is by chain rule (\itfacts{chain-rule}), and the inequality follows from Proposition~\ref{prop:info-decrease}, as we prove below that $\Prot_{1,\bid} \perp \Prot_1^{<\bid} \mid \Istar_r,\Sigma,J$ 
	or equivalently $\mi{\Prot_{1,\bid}}{\Prot_1^{<\bid} \mid \Istar_r, \Sigma,J} = 0$ (by \itfacts{info-zero}). 
	
	Define $\Ins:= \paren{I^{(1)},\ldots,I^{(n_r)}}$. As stated in Observation~\ref{obs:distr-input-independence}, we have $\Ins_\bid \perp \Ins^{<\bid}  \mid \Istar_r,\Sigma,J$ and
	 hence $\mi{\Ins_\bid}{\Ins^{<\bid} \mid \Istar_r,\Sigma,J} = 0$  (by \itfacts{info-zero}).
	 Moreover, notice that for any player $\bid \in N$, $\Prot_{1,\bid}$ is a deterministic function of $I^{(\bid)},\Sigma,J$ and hence, conditioned
	on $\Istar_r,\Sigma,J$, message $\Prot_{1,\bid}$ is only a function of $\Ins_\bid = I^{(\bid)}$. Consequently, by data processing inequality (\itfacts{data-processing}), 
	we also have $\mi{\Prot_{1,\bid}}{\Prot_1^{<\bid} \mid \Istar_r,\Sigma,J} = 0$.  
\end{proof}

In the next step, we use a direct-sum style argument to show that,

\begin{lemma}\label{lem:multi-direct-sum}
	For any $\bid \in N$, $\mi{\Istar_r}{\Prot_{1,\bid} \mid \Sigma,J} \leq {\card{\Prot_{1,\bid}}}/{p}$. 
\end{lemma}
\begin{proof}
	We have, 
	\begin{align*}
		\mi{\Istar_r}{\Prot_{1,\bid} \mid \Sigma,J} &= \Ex_{j \in [p]}\Bracket{\mi{\Istar_{r}}{\Prot_{1,\bid} \mid \Sigma,J=j}} = \Ex_{j \in [p]} \Bracket{\mi{I_{\bid,j}}{\Prot_{1,\bid} \mid \Sigma,J=j}} \tag{$\Istar_r = I_{\bid,j}$ conditioned
		 on $J=j$} \\
		&= \frac{1}{p} \cdot \sum_{j=1}^{p} \mi{I_{\bid,j}}{\Prot_{1,\bid} \mid \Sigma,J=j} \tag{the index $\jstar$ is chosen uniformly at random from $[p]$}
	\end{align*}

	Define $\Sigma^{-\bid}$ as the part of permutation $\Sigma$ that does not affect the labeling function $\Phi_\bid$ of player $\bid$, i.e., the values that are \emph{not} used to map the the input of player 
	$\bid$ (and all players in the same group as $\bid$) to $M$. With this notation, $\Sigma$ can be written as a function
	 of $\Phi_\bid$, $\Sigma^{-\bid}$, and $J$ (as $J$ and $\Phi_\bid$ uniquely define the rest of $\Sigma$ outside $\Sigma^{-\bid}$). 
	Consequently, we can write, 
	\begin{align*}
		\mi{\Istar_r}{\Prot_{1,\bid} \mid \Sigma,J} &= \frac{1}{p} \cdot \sum_{j=1}^{p} \mi{I_{\bid,j}}{\Prot_{1,\bid} \mid \Sigma^{-\bid},\Phi_\bid,J=j}
	\end{align*}
	
	Our goal is now to drop the conditioning on the event ``$J = j$''. To do so, notice that the distribution of $(\Sigma^{-\bid},\Phi_\bid)$ is independent of the event $J=j$; this is immediate to see as $\Sigma^{-\bid}$
	is independent of $\Phi_\bid$ and $J=j$, and $\Phi_\bid$ is independent of $J=j$ by Observation~\ref{obs:distr-j-independence}. Moreover, $I_{\bid,j}$ is independent of all $(\Sigma^{-\bid},\Phi_\bid,J=j)$ as 
	it is chosen independently from $\dist_{r-1}$ and furthermore, $\Prot_{1,\bid}$ is a function of $\Phi_\bid,I^{(\bid)}$, which are independent of $J=j$ (again by Observation~\ref{obs:distr-j-independence}). Consequently, we can
	drop the conditioning in the above equation and obtain that, 
	\begin{align*}
		\mi{\Istar_r}{\Prot_{1,\bid} \mid \Sigma,J} &= \frac{1}{p} \cdot \sum_{j =1}^{p} \mi{I_{\bid,j}}{\Prot_{1,\bid} \mid \Sigma^{-\bid},\Phi_\bid}\\
		&\leq \frac{1}{p} \cdot \sum_{j=1}^p \mi{I_{\bid,j}}{\Prot_{1,\bid} \mid I^{(\bid)<j},\Sigma^{-\bid},\Phi_\bid} \tag{by Proposition~\ref{prop:info-increase} as $I_{\bid,j} \perp I^{(\bid)<j} \mid \Sigma^{-\bid}, \Phi_\bid$} \\
		&= \frac{1}{p} \cdot \mi{I^{(\bid)}}{\Prot_{1,\bid} \mid \Sigma^{-\bid},\Phi_\bid} \leq \HH(\Prot_{1,\bid})/p \leq \card{\Prot_{1,\bid}}/p 
	\end{align*}
	where the equality is by chain rule (\itfacts{chain-rule}) and final inequality is by \itfacts{uniform}. 
\end{proof}

We now have, 
\begin{proof}[Proof of Lemma~\ref{lem:small-info}]
	By Claim~\ref{clm:multi-independent-message}, and Lemma~\ref{lem:multi-direct-sum},
	\begin{align*}
		\mi{\Istar_r}{\Prot_1 \mid \Sigma,J} \leq \sum_{\bid \in N} \mi{\Istar_r }{\Prot_{1,\bid} \mid \Sigma,J} \leq \sum_{\bid \in N} \card{\Prot_{1,\bid}}/p = \card{\Prot_1}/p = o(1/r^4)
	\end{align*}
	by the lemma assumption that $\card{\Prot_1} = o(p/r^4)$. 
\end{proof}

Recall that $\Istar_r$ is the special instance in distribution $\dist_r$ which was sampled from distribution $\dist_{r-1}$. 
We define $\psi_r$ as the distribution of $\Istar_r$ conditioned on $(\Prot_1,\Sigma,J)$, i.e., 
after seeing the first message of $\prot$ and  the easy part of the input $(\Sigma,J)$. 
As a corollary of Lemma~\ref{lem:small-info}, we have that this further conditioning does not change the distribution of $\Istar_r$ by much. 

\begin{claim}\label{clm:small-distance}
	If $\card{\Prot_r} = o({p}/{r^4})$, then, $\Ex_{(\Prot_1 ,\Sigma,J)}\Bracket{\tvd{\psi_r}{\dist_{r-1}}} = o(1/r^2)$. 
\end{claim}
\begin{proof}
	We have, 
	\begin{align*}
		\Ex_{(\Prot_1,\Sigma,J)}\Bracket{\tvd{\psi_r}{\dist_{r-1}}} &= \Ex_{(\Prot_1,\Sigma,J)}\Bracket{\tvd{\psi_r}{\distribution{\Istar_r \mid (\Sigma,J)}}}\tag{$\distribution{\Istar_r} = \dist_{r-1}$ and $\Istar_r \perp \Sigma,J$} \\
		&\leq \Ex_{(\Prot_1 , \Sigma,J)}\Bracket{\sqrt{\frac{1}{2} \cdot \DD{\psi_r}{\distribution{\Istar_r \mid (\Sigma,J)} }}} \tag{by Pinsker's inequality (Fact~\ref{fact:pinskers})} \\
		&\leq \sqrt{\frac{1}{2} \cdot \Ex_{(\Prot_1 ,\Sigma,J)}\Bracket{\DD{\psi_r}{\distribution{\Istar_r \mid (\Sigma,J)} }}} \tag{by concavity of $\sqrt{\cdot}$ and Jensen's inequality} \\
		&= \sqrt{\frac{1}{2} \cdot \mi{\Istar_r}{\Prot_1 \mid \Sigma,J}}  \tag{by Fact~\ref{fact:kl-info}}
	\end{align*}
	which is $o(1/r^2)$ by Lemma~\ref{lem:small-info}.
\end{proof}

We are now ready to state the main result of this section. Define the recursive function $e(r):= e(r-1) + o(1/r^2)$ (with $e(0) = 0$).
Note that $e(r) = \sum_{i=1}^{r}o(1/i^2) = o(1)$. We have, 

\begin{lemma}\label{lem:embedding}
	For any $r \geq 1$, any $r$-round protocol $\prot$ for determining $\thetas$ on $\dist_r$ with error probability at most $\delta = 1/3-e(r)$ 
	requires $\Omega(p/r^4)$ communication. 
\end{lemma}
\begin{proof}
	We prove this lemma inductively. The base case for $r=1$ follows from Lemma~\ref{lem:sim-conclude}. Now suppose the result holds for all integers smaller than $r$ and we aim to prove it for the 
	case of $r$-round protocols. Let $\prot$ be a $\delta$-error protocol for estimating $\thetas$ with $\delta = 1/3 - e(r)$ and assume by contradiction that the communication cost of $\prot$ is $o(p/r^4)$; we use $\prot$ to design 
	a randomized $(r-1)$-round protocol $\prot'$ that has communication cost $o(p/r^4)$, and errs w.p. at most $1/3-e(r-1)$ on $\dist_{r-1}$, and then use averaging argument to fix its randomness
	to obtain a deterministic protocol that contradicts the induction hypothesis.

\textbox{Protocol $\prot'$: \textnormal{An $(r-1)$-round protocol for solving instances of $\dist_{r-1}$ using protocol $\prot$.}}{

\smallskip

\textbf{Input:} An instance $I \sim \dist_{r-1}$.  \textbf{Output:} The value of $\thetas$ in $I$. 
\\
\algline

\begin{enumerate}
	\item Let $N = [n_{r}]$ and $M = [m_r]$. 
	\item \label{line:sampled} Using \emph{public randomness}, the players sample $(\Prot_1,\sigma,\jstar) \sim \dist_r(N,M)$, i.e., they sample from the joint distribution of the first message of $\prot$ (denoted by $\Prot_1$), 
	the permutation $\sigma$ over $M$, and the index $\jstar \in [p]$. 
	\item The players partition $N$ into $k^2$ equal-size groups $\FN = (N_1,\ldots,N_{k^2})$ (as is done in $\dist_r$) and the $i$-th player (denoted by $P_i$) in $I$ mimics the role
	 of the $i$-th player in each group $N_g \in \FN$ (denoted by $P_{i,g}$) 
	individually, as follows: 
	\begin{enumerate}
		\item $P_i$ sets the input for $P_{i,g}$ (for $g \in [k^2]$) in the instance $I_{i,\jstar}$ (in $\dist_r$) as the input of $P_i$ in the input instance $I$ mapped via $\sigma$ to $M$ (using the same procedure as in $\dist_r$). 
		\item \label{line:valid} $P_i$ samples the input for $P_{i,g}$ (for $g \in [k^2]$) in all other instances $I_{i,j}$ (for $j \neq \jstar$), using \emph{private randomness} from
		the distribution $I_{i,j} \sim \dist_{r} \mid (\Istar=I,\Prot_1,\sigma,\jstar)$ (we prove this is indeed possible by Proposition~\ref{prop:private-sampling} below).	
	\end{enumerate}
	\item The players run the protocol $\prot$ on the new sampled instance conditioned on the first message being $\Prot_1$, (i.e., run $\prot$ from the second round
	assuming $\Prot_1$ is the content of blackboard after the first round) and output the same answer as $\prot$. 
\end{enumerate}
}
We start by arguing that $\prot'$ is indeed a valid protocol; in particular, Line~(\ref{line:valid}) can be implemented without any communication. 
We first need some new notation. For any player $\bid \in N$, define $\Istar_r(\bid)$ as the input of player $\bid$ in the instance $I_{\bid,\jstar} = \Istar_r$ (conditioned on $\Sigma,J$),
and define $\Istar_r(-\bid)$ as the input of all other players in $\Istar_r$. 
To prove that $\prot'$ is valid, it suffices to prove the following proposition. 

\begin{proposition}\label{prop:private-sampling} 
	The distribution $\Ins := (I^{(1)},\ldots,I^{(n)}) \sim \paren{\dist_{r} \mid \Istar_r,\Prot_1,\Sigma,J}$ is a \emph{product distribution} whereby 
	each $\Ins_\bid = I^{(\bid)}$ is sampled from $\dist_{r} \mid \Istar_r(\bid) , \Prot_1,\Sigma,J$. 
\end{proposition}
\begin{proof}
	For any $\bid \in N$, we prove that $\mi{\Ins_\bid}{\Ins^{-\bid},\Istar_r(-\bid) \mid  \Istar_r(\bid),\Prot_1,\Sigma,J} = 0$. By \itfacts{info-zero}, this implies
	 that $\Ins_\bid \perp (\Ins^{-\bid},\Istar_r(-\bid)) \mid  \Istar_r(\bid),\Prot_1,\Sigma,J$, hence proving the proposition. We have,
	\begin{align*}
		\mi{\Ins_\bid}{\Ins^{-\bid},\Istar_r(-\bid) \mid  \Istar_r(\bid),\Prot_1,\Sigma,J} &=\mi{\Ins_\bid}{\Ins^{-\bid},\Istar_r(-\bid) \mid  \Istar_r(\bid),\Prot_{1,\bid},\Prot^{-\bid}_1,\Sigma,J} 
		\tag{as $\Prot_1 = \Prot_{1,\bid},\Prot_1^{-\bid}$}\\
		&\leq \mi{\Ins_\bid}{\Ins^{-\bid},\Istar_r(-\bid) \mid  \Istar_r(\bid),\Prot_{1,\bid},\Sigma,J}
	\end{align*}
	since $\Ins_\bid \perp \Prot_1^{-\bid} \mid  (\Ins^{-\bid},\Istar_r=\paren{\Istar_r(\bid),\Istar_r(-\bid)},\Prot_{1,\bid},\Sigma,J)$ as $\Prot_1^{-\bid}$ is a deterministic function of $\Ins^{-\bid},\Istar_r,\Sigma,J$, and hence 
	we can apply Proposition~\ref{prop:info-decrease}.  
	
	Furthermore, 
	\begin{align*}
		\mi{\Ins_\bid}{\Ins^{-\bid},\Istar_r(-\bid) \mid  \Istar_r(\bid),\Prot_1,\Sigma,J} &\leq \mi{\Ins_\bid}{\Ins^{-\bid},\Istar_r(-\bid) \mid  \Istar_r(\bid),\Prot_{1,\bid},\Sigma,J} \\
		&\leq \mi{\Ins_\bid}{ \Ins^{-\bid},\Istar_r(-\bid) \mid  \Istar_r(\bid),\Sigma,J} 
	\end{align*}
	since $\paren{\Ins^{-\bid},\Istar_r(-\bid)} \perp \Prot_{1,\bid} \mid  \Ins_\bid,\Istar_r(\bid),\Sigma,J)$ as $\Prot_{1,\bid}$ is a deterministic function of $I^{(\bid)}=(\Ins_\bid,\Istar_r(\bid)),\Sigma,J$ and hence we can again apply
	Proposition~\ref{prop:info-decrease}. Finally, $\mi{\Ins_\bid}{ \Ins^{-\bid},\Istar_r(-\bid) \mid  \Istar_r(\bid),\Sigma,J}  = 0$ by Observation~\ref{obs:distr-input-independence} and \itfacts{info-zero}, implying 
	that $\mi{\Ins_\bid}{\Ins^{-\bid},\Istar_r(-\bid) \mid  \Istar_r(\bid),\Prot_1,\Sigma,J} = 0$ as well, proving the proposition.  
\end{proof}

It is now easy to see that $\prot'$ is indeed an $(r-1)$-round protocol: to sample from the distribution $\dist_r \mid (\Istar = I,\Prot_1,\Sigma,J)$ in Line~(\ref{line:valid}), each player $\bid \in N$ needs to 
sample from the distribution $\dist_{r} \mid (\Istar_r(\bid), \Prot_1,\Sigma,J)$ (by Proposition~\ref{prop:private-sampling}), and this is possible since $(\Istar_r(\bid), \Prot_1,\Sigma,J)$ are all known to $\bid$. 
Hence, the players do not need any communication for simulating the first round of protocol $\prot$. We now prove that.
\begin{claim}\label{clm:prot'-correct}
 $\prot'$ is a $\delta'$-error protocol for $\dist_{r-1}$ for $\delta' = 1/3 - e(r-1)$. 
\end{claim}
\begin{proof}
	Note that our goal is to calculate the probability that $\prot'$ errs given an instance $I \sim \dist_{r-1}$. 
	For the sake of analysis, suppose  that $I \sim \psi_r$ instead, i.e., is sampled from the distribution $\distribution{\Istar_r \mid \Prot_1, \Sigma,J}$ (according to distribution $\dist_r$). 
	In this case, one can see that the distribution of the $r$-round instance constructed by $\prot'$ matches the distribution $\dist_r$. Since $\prot'$ outputs the same answer as $\prot$ on this new sampled instance, 
	and since $I = \Istar_r$ in the new instance, the probability that $\prot'$ errs on $\psi_r$ is equal to the probability that $\prot$ errs on $\dist_r$ which in turn is equal to $1/3 - e(r)$. Now notice that by 
	Claim~\ref{clm:small-distance}, the total variation distance between $\psi_r$ and $\dist_{r-1}$ is $o(1/r^2)$ and hence by Fact~\ref{fact:tvd-small}, $\Pr_{\dist_{r-1}}\paren{\prot'~\errs} \leq \Pr_{\psi_r}\paren{\prot'~\errs} + 
	o(1/r^2) = 1/3 - e(r-1)$. 
	
	We now formalize the above intuition. Define $\Rpri$ and $\Rpub$ as, respectively, the private and the public randomness of protocol $\prot'$. 
	The probability that $\prot'$ errs on an instance $I \sim \dist_{r-1}$ can be written as, 
	\begin{align*}
	\Pr_{\dist_{r-1}}\paren{\prot'~\errs} &=\Ex_{I \sim \dist_{r-1}}  \Ex_{\Rpub} \Bracket{\Pr_{\Rpri}\paren{\prot'~\errs \mid \Rpub}} \\
	&= \Ex_{\Rpub}~  \Ex_{I \sim \dist_{r-1} \mid \Rpub} \Bracket{\Pr_{\Rpri}\paren{\prot'~\errs \mid \Rpub}} \tag{as $\Rpub \perp I$} \\
	&= \Ex_{(\Prot_1,\Sigma,J)}~  \Ex_{I \sim \dist_{r-1} \mid (\Prot_1,\Sigma,J)} \Bracket{\Pr_{\Rpri}\paren{\prot'~\errs \mid \Prot_1,\Sigma,J}} \tag{as $\Rpub = (\Prot_1,\Sigma,J)$} \\
	&\leq \Ex_{(\Prot_1,\Sigma,J)}\Bracket{\Ex_{I \sim \psi_i \mid (\Prot_1,\Sigma,J)} \Bracket{\Pr_{\Rpri}\paren{\prot'~\errs \mid \Prot_1,\Sigma,J}} + \tvd{\dist_{r-1}}{\paren{\psi_i \mid (\Prot_1,\Sigma,J)}}} 
	\tag{by Fact~\ref{fact:tvd-small}} \\
	&=  \Ex_{(\Prot_1,\Sigma,J)}\Bracket{\Ex_{I \sim \psi_i \mid (\Prot_1,\Sigma,J)} \Bracket{\Pr_{\Rpri}\paren{\prot'~\errs \mid \Prot_1,\Sigma,J}}} + \Ex\Bracket{\tvd{\dist_{r-1}}{{\psi_i}}} \\
	&=  \Ex_{(\Prot_1,\Sigma,J)}\Bracket{\Ex_{I \sim \psi_i \mid (\Prot_1,\Sigma,J)} \Bracket{\Pr_{\Rpri}\paren{\prot'~\errs \mid \Prot_1,\Sigma,J}}} +o(1/r^2) \tag{by Claim~\ref{clm:small-distance}} \\
	&= \Ex_{(\Prot_1,\Sigma,J)}\Bracket{\Ex_{I \sim \psi_i \mid (\Prot_1,\Sigma,J)} \Bracket{\Pr_{\dist_r}\paren{\prot'~\errs \mid \Istar_r=I,\Prot_1,\Sigma,J}}} +o(1/r^2) \tag{the distribution of sampled instances in $\prot'$ 
	(via $\Rpri$) matches $\dist_r \mid \Istar = I, \Prot_1,\Sigma,J$} \\
	&= \Ex_{(\Prot_1,\Sigma,J)}\Bracket{\Ex_{I \sim \psi_i \mid (\Prot_1,\Sigma,J)} \Bracket{\Pr_{\dist_r}\paren{\prot~\errs \mid \Istar_r=I,\Prot_1,\Sigma,J}}} +o(1/r^2)  \tag{the output of $\prot'$ and $\prot$ is the same} \\
	&= \Ex_{(\Istar_r,\Prot_1,\Sigma,J)}\Bracket{\Pr_{\dist_r}\paren{\prot~\errs \mid \Istar_r,\Prot_1,\Sigma,J}} +o(1/r^2) \tag{$\psi_i = \distribution{\Istar_r \mid \Prot_1,\Sigma,J}$ by definition} \\
	&= \Pr_{\dist_{r}}\paren{\prot~\errs} + o(1/r^2) = 1/3-e(r) + o(1/r^2) = 1/3 - e(r-1)
	\end{align*}
	finalizing the proof.
\end{proof}

Lemma~\ref{lem:embedding} now follows from Claim~\ref{clm:prot'-correct} by an averaging argument since we can fix the randomness in $\prot'$ to obtain a deterministic protocol $\prot''$ that uses 
$o(p/r^4)$ bits of communication and errs w.p. at most $1/3 - e(r-1)$ on $\dist_{r-1}$, a contradiction with the induction hypothesis. 
\end{proof}

Theorem~\ref{thm:multi-lower} now easily follows from Lemma~\ref{lem:embedding}.

\begin{proof}[Proof of Theorem~\ref{thm:multi-lower}]
	Let $\prot$ be a $\paren{\frac{1}{r}\cdot m^{\frac{1-2\eps}{2r+1}}}$-approximation, $(1/4)$-error protocol for subadditive combinatorial auctions on the distribution $\dist_{r}$.  
	By Claim~\ref{clm:multi-theta-apx}, $\prot$ is also a $\paren{1/4+o(1)}$-error protocol for $\thetas$ estimation on $\dist_{r}$. Since $\paren{1/4+o(1)} < 1/3-e(r)$, 
	by Lemma~\ref{lem:embedding}, we have $\norm{\prot} = \Omega(p/r^4) = \exp\paren{\Theta(k^{\eps})}/r^4 = \exp\paren{m^{\Omega(\eps/r)}}/r^4 = \exp\paren{m^{\Omega(\eps/r)}}$,
	as $k = m^{\Omega(1/r)}$ and $r = o(\frac{\log{m}}{\log\log{m}})$. Re-parametrizing $\eps$ by $\eps/2$ in the lower bound argument finalizes the proof. 
\end{proof}

%% file: conclusion.tex
\section{Conclusion}\label{sec:conclusion} 
In this paper, we studied the role of interaction in obtaining efficient allocations in subadditive combinatorial auctions. We showed that
for any $r \geq 1$, any $r$-round protocol that uses polynomial communication can only achieve an $\Omega(\frac{1}{r} \cdot m^{1/2r+1})$ approximation to the 
optimal social welfare. This settles an open question posed by Dobzinski~\etal~\cite{DobzinskiNO14} and Alon~\etal~\cite{AlonNRW15} on the round-approximation tradeoff of polynomial communication protocols in subadditive combinatorial auctions. 

An immediate corollary of our main result is that $\Omega(\frac{\log{m}}{\log\log{m}})$ rounds of interaction
are necessary for obtaining an efficient allocation with polynomial communication in subadditive combinatorial auctions. The qualitative message of this theoretical
result is that \emph{a modest amount of interaction between individuals in a market is crucial for obtaining an efficient allocation}. This further support
the point of view of~\cite{DobzinskiNO14} on the necessity of interaction for economic efficiency. 

An interesting direction for future research, also advocated by~\cite{DobzinskiNO14}, is to consider the case where the bidders valuations are \emph{submodular}. 
It is known that obtaining a better than $(1-1/2e)$-approximation to social welfare in submodular combinatorial auctions requires exponential communication~\cite{DobzinskiV13} (regardless of the number of rounds of interaction). 
However, no better lower bounds are known for bounded-round protocols (even for simultaneous ones). 
Another interesting open problem is to close the gap between the $\Omega(\log\log{n})$ lower bound of~\cite{AlonNRW15} and 
the $O(\log{n})$ upper bound of~\cite{DobzinskiNO14} on the number of rounds necessary to achieve an efficient allocation in matching markets.

%% file: main.bbl
\begin{thebibliography}{10}

\bibitem{AharoniEL88}
R.~Aharoni, P.~Erd{\"{o}}s, and N.~Linial.
\newblock Optima of dual integer linear programs.
\newblock {\em Combinatorica}, 8(1):13--20, 1988.

\bibitem{AlonNRW15}
N.~Alon, N.~Nisan, R.~Raz, and O.~Weinstein.
\newblock Welfare maximization with limited interaction.
\newblock In {\em {IEEE} 56th Annual Symposium on Foundations of Computer
  Science, {FOCS} 2015, Berkeley, CA, USA, 17-20 October, 2015}, pages
  1499--1512, 2015.

\bibitem{AssadiKL17}
S.~Assadi, S.~Khanna, and Y.~Li.
\newblock On estimating maximum matching size in graph streams.
\newblock In {\em Proceedings of the Twenty-Eighth Annual {ACM-SIAM} Symposium
  on Discrete Algorithms, {SODA} 2017, Barcelona, Spain, Hotel Porta Fira,
  January 16-19}, pages 1723--1742, 2017.

\bibitem{BadanidiyuruDFKNR12}
A.~Badanidiyuru, S.~Dobzinski, H.~Fu, R.~Kleinberg, N.~Nisan, and
  T.~Roughgarden.
\newblock Sketching valuation functions.
\newblock In {\em Proceedings of the Twenty-Third Annual {ACM-SIAM} Symposium
  on Discrete Algorithms, {SODA} 2012, Kyoto, Japan, January 17-19, 2012},
  pages 1025--1035, 2012.

\bibitem{BalcanCIW12}
M.~Balcan, F.~Constantin, S.~Iwata, and L.~Wang.
\newblock Learning valuation functions.
\newblock In {\em {COLT} 2012 - The 25th Annual Conference on Learning Theory,
  June 25-27, 2012, Edinburgh, Scotland}, pages 4.1--4.24, 2012.

\bibitem{BravermanMW17}
M.~Braverman, J.~Mao, and S.~M. Weinberg.
\newblock On simultaneous two-player combinatorial auctions.
\newblock {\em arXiv preprint arXiv:1704.03547}, 2017.

\bibitem{BriestKV11}
P.~Briest, P.~Krysta, and B.~V{\"{o}}cking.
\newblock Approximation techniques for utilitarian mechanism design.
\newblock {\em {SIAM} J. Comput.}, 40(6):1587--1622, 2011.

\bibitem{ITbook}
T.~M. Cover and J.~A. Thomas.
\newblock {\em Elements of information theory {(2.} ed.)}.
\newblock Wiley, 2006.

\bibitem{DobzinskiNO14}
S.~Dobzinski, N.~Nisan, and S.~Oren.
\newblock Economic efficiency requires interaction.
\newblock In {\em Symposium on Theory of Computing, {STOC} 2014, New York, NY,
  USA, May 31 - June 03, 2014}, pages 233--242, 2014.

\bibitem{DobzinskiNS05}
S.~Dobzinski, N.~Nisan, and M.~Schapira.
\newblock Approximation algorithms for combinatorial auctions with
  complement-free bidders.
\newblock In {\em Proceedings of the 37th Annual {ACM} Symposium on Theory of
  Computing, Baltimore, MD, USA, May 22-24, 2005}, pages 610--618, 2005.

\bibitem{DobzinskiS06}
S.~Dobzinski and M.~Schapira.
\newblock An improved approximation algorithm for combinatorial auctions with
  submodular bidders.
\newblock In {\em Proceedings of the Seventeenth Annual {ACM-SIAM} Symposium on
  Discrete Algorithms, {SODA} 2006, Miami, Florida, USA, January 22-26, 2006},
  pages 1064--1073, 2006.

\bibitem{DobzinskiV13}
S.~Dobzinski and J.~Vondr{\'{a}}k.
\newblock Communication complexity of combinatorial auctions with submodular
  valuations.
\newblock In {\em Proceedings of the Twenty-Fourth Annual {ACM-SIAM} Symposium
  on Discrete Algorithms, {SODA} 2013, New Orleans, Louisiana, USA, January
  6-8, 2013}, pages 1205--1215, 2013.

\bibitem{ConcentrationBook}
D.~P. Dubhashi and A.~Panconesi.
\newblock {\em Concentration of Measure for the Analysis of Randomized
  Algorithms}.
\newblock Cambridge University Press, 2009.

\bibitem{DuttingK17}
P.~D{\"{u}}tting and T.~Kesselheim.
\newblock Best-response dynamics in combinatorial auctions with item bidding.
\newblock In {\em Proceedings of the Twenty-Eighth Annual {ACM-SIAM} Symposium
  on Discrete Algorithms, {SODA} 2017, Barcelona, Spain, Hotel Porta Fira,
  January 16-19}, pages 521--533, 2017.

\bibitem{Feige09}
U.~Feige.
\newblock On maximizing welfare when utility functions are subadditive.
\newblock {\em {SIAM} J. Comput.}, 39(1):122--142, 2009.

\bibitem{FeigeV06}
U.~Feige and J.~Vondr{\'{a}}k.
\newblock Approximation algorithms for allocation problems: Improving the
  factor of 1 - 1/e.
\newblock In {\em 47th Annual {IEEE} Symposium on Foundations of Computer
  Science {(FOCS} 2006), 21-24 October 2006, Berkeley, California, USA,
  Proceedings}, pages 667--676, 2006.

\bibitem{FeldmanFGL13}
M.~Feldman, H.~Fu, N.~Gravin, and B.~Lucier.
\newblock Simultaneous auctions are (almost) efficient.
\newblock In {\em Symposium on Theory of Computing Conference, STOC'13, Palo
  Alto, CA, USA, June 1-4, 2013}, pages 201--210, 2013.

\bibitem{FeldmanGL15}
M.~Feldman, N.~Gravin, and B.~Lucier.
\newblock Combinatorial auctions via posted prices.
\newblock In {\em Proceedings of the Twenty-Sixth Annual {ACM-SIAM} Symposium
  on Discrete Algorithms, {SODA} 2015, San Diego, CA, USA, January 4-6, 2015},
  pages 123--135, 2015.

\bibitem{ImpagliazzoK10}
R.~Impagliazzo and V.~Kabanets.
\newblock Constructive proofs of concentration bounds.
\newblock In {\em Approximation, Randomization, and Combinatorial Optimization.
  Algorithms and Techniques, 13th International Workshop, {APPROX} 2010, and
  14th International Workshop, {RANDOM} 2010, Barcelona, Spain, September 1-3,
  2010. Proceedings}, pages 617--631, 2010.

\bibitem{LaviS11}
R.~Lavi and C.~Swamy.
\newblock Truthful and near-optimal mechanism design via linear programming.
\newblock {\em J. {ACM}}, 58(6):25:1--25:24, 2011.

\bibitem{LehmannLN06}
B.~Lehmann, D.~J. Lehmann, and N.~Nisan.
\newblock Combinatorial auctions with decreasing marginal utilities.
\newblock {\em Games and Economic Behavior}, 55(2):270--296, 2006.

\bibitem{Lovasz75}
L.~Lov{\'a}sz.
\newblock On the ratio of optimal integral and fractional covers.
\newblock {\em Discrete mathematics}, 13(4):383--390, 1975.

\bibitem{MiltersenNSW98}
P.~B. Miltersen, N.~Nisan, S.~Safra, and A.~Wigderson.
\newblock On data structures and asymmetric communication complexity.
\newblock {\em J. Comput. Syst. Sci.}, 57(1):37--49, 1998.

\bibitem{Nisan02}
N.~Nisan.
\newblock The communication complexity of approximate set packing and covering.
\newblock In {\em Automata, Languages and Programming, 29th International
  Colloquium, {ICALP} 2002, Malaga, Spain, July 8-13, 2002, Proceedings}, pages
  868--875, 2002.

\bibitem{NisanS06}
N.~Nisan and I.~Segal.
\newblock The communication requirements of efficient allocations and
  supporting prices.
\newblock {\em J. Economic Theory}, 129(1):192--224, 2006.

\bibitem{PanconesiS97}
A.~Panconesi and A.~Srinivasan.
\newblock Randomized distributed edge coloring via an extension of the
  chernoff-hoeffding bounds.
\newblock {\em {SIAM} J. Comput.}, 26(2):350--368, 1997.

\bibitem{SenV08}
P.~Sen and S.~Venkatesh.
\newblock Lower bounds for predecessor searching in the cell probe model.
\newblock {\em J. Comput. Syst. Sci.}, 74(3):364--385, 2008.

\bibitem{Vondrak08}
J.~Vondr{\'{a}}k.
\newblock Optimal approximation for the submodular welfare problem in the value
  oracle model.
\newblock In {\em Proceedings of the 40th Annual {ACM} Symposium on Theory of
  Computing, Victoria, British Columbia, Canada, May 17-20, 2008}, pages
  67--74, 2008.

\bibitem{Yao79}
A.~C. Yao.
\newblock Some complexity questions related to distributive computing
  (preliminary report).
\newblock In {\em Proceedings of the 11h Annual {ACM} Symposium on Theory of
  Computing, April 30 - May 2, 1979, Atlanta, Georgia, {USA}}, pages 209--213,
  1979.

\end{thebibliography}
